\documentclass[onefignum,onetabnum
]{siamart171218}

\usepackage{lipsum}
\usepackage{amsfonts}
\usepackage{graphicx}
\usepackage{epstopdf}
\usepackage{amssymb}
\usepackage{algorithmic}
\ifpdf
  \DeclareGraphicsExtensions{.eps,.pdf,.png,.jpg}
\else
  \DeclareGraphicsExtensions{.eps}
\fi

\newsiamremark{remark}{Remark}
\newsiamremark{hypothesis}{Hypothesis}
\crefname{hypothesis}{Hypothesis}{Hypotheses}
\newsiamthm{claim}{Claim}

\headers{Discrete Breathers in a Honeycomb Lattice Near a Semi-Dirac Point}{A. Hofstrand}

\title{Discrete Breathers in a Honeycomb Lattice Near a Semi-Dirac Point}

\author{Andrew Hofstrand\thanks{Department of Mathematics, New York Institute of Technology, New York, NY (\email{ahofstra@nyit.edu}).}
}

\usepackage{amsopn}

\makeatletter
\newcommand*{\addFileDependency}[1]{
  \typeout{(#1)}
  \@addtofilelist{#1}
  \IfFileExists{#1}{}{\typeout{No file #1.}}
}
\makeatother

\newcommand*{\myexternaldocument}[1]{%
    \externaldocument{#1}%
    \addFileDependency{#1.tex}%
    \addFileDependency{#1.aux}%
}

\myexternaldocument{ex_supplement}
\usepackage{comment}

\begin{document}

\maketitle

\begin{abstract} 
We study the dynamics of discrete breathers---spatially localized and time-periodic solutions---inside the bandgap of a nonlinear honeycomb lattice where the dispersion landscape approaches a so-called semi-Dirac point in which the bands cross linearly in one direction and quadratically in the orthogonal direction.  By studying breather dynamics in two opposing asymptotic regimes, near the continuum and anti-continuum limits, we capture the spatial profiles of hybrid coherent structures having central cores supported on a finite number of lattice sites and infinite decaying tails that are well approximated by exact separable solutions to an effective long-wave PDE theory at spatial infinity.  We find that these breathers are dynamically stable over a wide range of parameters and find an instability transition. Finally, we analyze the Floquet stability of spatially extended nonlinear plane waves bifurcating from the zero solution at the edges of the gap and how they shape breather profiles inside the gap.          
\end{abstract}



\section{Introduction}
\subsection{Motivation, background, and outline}  With the advent of realizable 2D materials such as graphene~\cite{novo}---a single layer of carbon atoms arranged in an approximate honeycomb array---there has been an outpouring of theoretical and experimental studies on wave transport and localization properties in two-dimensional lattices. The bipartite honeycomb geometry frequently appears in applications in condensed matter systems~\cite{kane}, ultra-cold atoms in optical lattices~\cite{ess,haller}, coupled waveguide arrays~\cite{ma}, acoustic structures~\cite{hibbins}, and mechanical metamaterials~\cite{bertoldi}.  Many of these engineered platforms exhibit malleable nonlinear responses at sufficiently high excitation levels. Understanding the dynamics of energy localization in discrete nonlinear systems is of fundamental importance in their design. Here, we consider a well-known class of spatially localized and time-periodic solutions called \textit{discrete breathers}~\cite{flasch08} in the honeycomb lattice. Most of the literature on discrete breathers concerns one-dimensional lattices or rectangular lattices of higher dimensions.  However, several studies on the properties of discrete breathers in hexagonal and honeycomb lattices are given in \cite{kouk,kouk2,dai,wattis,palermo}.

Recently, there has been a renewed interest from the condensed matter community in a class of quasi-particle known as the \textit{semi-Dirac fermion}. The semi-Dirac fermion has the peculiar property that its energy dispersion relation is \underline{linear} in one direction and \underline{quadratic} in the orthogonal direction. This implies that the quasi-particle is \textit{massive} moving in one direction, $E(\mathbf{k}=(k_x,0))=\hbar^2k_x^2/2m$, and \textit{massless} moving in the perpendicular direction, $E(\mathbf{k}=(0,k_y))=\hbar vk_y$, where $E$ is energy, $\mathbf{k}$ is quasi-momentum, $m$ is mass, $\hbar$ is Planck's constant, and $v$ is the quasi-particle's effective velocity.  The semi-Dirac fermion was theoretically predicted to exist in a graphene-like structure in 2008~\cite{dietl}.  In 2024, it was experimentally observed for the first time in the topological metal ZrSiS using spectroscopy~\cite{basov}.

In the current work, we study the properties of discrete breathers in a conservative mass-spring system arranged in a honeycomb lattice with an in-cell cubic type nonlinearity.  The dimerized linear couplings of the sites in the lattice we consider are analogous to those of the one-dimensional Su-Schrieffer-Heeger (SSH) model for polyacetylene\cite{ssh79}.  Nonlinear extensions of the SSH model are analyzed in \cite{Hof:23, Li:24}. Here, we focus on discrete breathers that have a frequency located inside the honeycomb lattice's phonon bandgap where the dispersion bands approach the energy landscape of a semi-Dirac fermion.   Our model is described in the following subsection.

In Section \ref{sec:pde}, we derive an effective PDE description for the envelopes of weakly nonlinear wavepackets with wavevectors concentrated around a critical point of the lattice's band structure where the bands nearly meet in a semi-Dirac crossing.  We then obtain exact expressions for line solitons to the 2D PDE model, exponentially localized along the spatial $Z$ direction and constant along the transverse spatial $H$ direction.  Furthermore, using these expressions, we derive spatially separable asymptotic solutions in the limit $Z\to\pm\infty$ that are also exponentially localized in $H$.  These separable solutions can be formally expressed as the product of two exact one-dimensional solitons.

Section \ref{sec:anti} specifies conditions for the existence of discrete breathers in the lattice and provides an iterative scheme to compute their spatiotemporal profiles along with their dynamic stability from exact seed solutions. We find that the asymptotic envelopes derived in Section \ref{sec:pde} accurately represent the extended spatial tails of the computed breathers near the semi-Dirac crossing.  Importantly, our combination of analytic and numerical methods near opposing limits (the continuum and anti-continuum limits) allows us to fully describe breather characteristics not captured by either asymptotic theory alone. 

Finally, in Section \ref{sec:plane} we study de-localized nonlinear plane waves that bifurcate from the trivial solution at the edges of the phonon gap and analyze their Floquet stability both perturbatively and numerically. Section \ref{sec:con} gives some concluding remarks.

\subsection{Lattice model and band structure}
\label{sec:model}
We consider a two-dimensional Hamiltonian mass-spring system arranged in a regular honeycomb pattern (see Figure \ref{fig1}(a)) with linear nearest-neighbor couplings.  The unit cell of the lattice consists of two sub-sites, labeled $A$ and $B$, with identical masses set to one. Here, we fix the \textit{intra-cell} coupling strength to unity and allow the \textit{inter-cell} coupling strength, denoted by $\lambda$, to vary in the interval $0\leq\lambda\leq 1$.

 \begin{figure}[h!]
  \centering
  \includegraphics[scale=0.55]{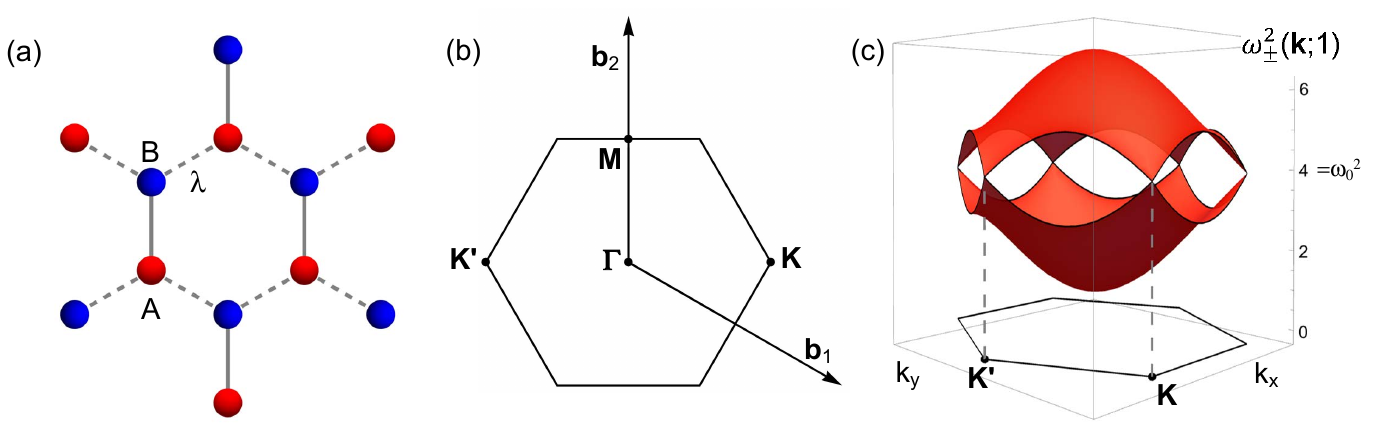}
  \caption{(a) schematic of honeycomb lattice with A (red) and B (blue) sites and staggered coupling strengths 1 (bold) and $\lambda$ (dashed); (b) the first Brillouin zone, $\mathcal{B},$ and reciprocal lattice vectors; (c) the linear dispersion relation, $\omega_{\pm}^2(\mathbf{k};\lambda),$ in \eqref{eq2} over $\mathcal{B}$ for the case $\lambda=1$ and $\omega_0=2$.}
  \label{fig1}
\end{figure}

The family of dynamical equations on the lattice is given by 
\begin{align}\label{eq1}
&\ddot{x}_{n,m}^A=-\partial_{z_1}V(x_{n,m}^A,x_{n,m}^B)+\lambda\left[x_{n,m-1}^B+x_{n+1,m-1}^B\right]\\ \nonumber
&\ddot{x}_{n,m}^B=-\partial_{z_2}V(x_{n,m}^A,x_{n,m}^B)+\lambda\left[x_{n,m+1}^A+x_{n-1,m+1}^A\right],
\end{align}
where $(n,m)\in\mathbb{Z}^2$, $x_{n,m}(t)$ is the displacement of the mass of the $(n,m)$ point in the direction perpendicular to the lattice, and $V\in C^2(\mathbb{R}^2)$ is an in-cell potential. We assume that $V$ is a convex function near $(0,0)$.

Specifically, we consider the potential form,
\begin{equation}
\label{potential}
    V(z_1,z_2)=\dfrac{1}{2}\omega_0^2\left(z_1^2+z_2^2\right)-z_1z_2+
    \dfrac{1}{4}g\left(z_1^4+2z_1^2z_2^2+z_2^4\right), 
\end{equation} 
where $\omega_0^2>3$ is the on-site harmonic oscillator frequency and $g=\pm 1$ represents a \textit{hardening} or \textit{softening} nonlinearity, respectively.  

We fix the spatial scale in the plane of the lattice so that the nearest-neighbor distance is unity.  The infinite lattice is then generated by all integer linear combinations of the vectors 
\[\mathbf{a}_1=\sqrt{3}(1,0)^{\top}\text{ and } \mathbf{a}_2=\dfrac{\sqrt{3}}{2}(1,\sqrt{3})^{\top}.\]

Similarly, the reciprocal lattice, given in terms of the wavevector $\mathbf{k}=(k_x,k_y)^{\top}\in\mathbb{R}^2$, is spanned by integer multiples of the vectors
\[\mathbf{b}_1=2\pi\left(\dfrac{1}{\sqrt{3}},-\dfrac{1}{3}\right)^{\top} \text{ and } \mathbf{b}_2=2\pi\left(0,\dfrac{2}{3}\right)^{\top}.\] The hexagonal primitive cell or the first Brillouin zone, $\mathcal{B}$, of the reciprocal lattice is shown in Figure \ref{fig1}(b). $\mathcal{B}$ contains the labeled critical wavevectors $\mathbf{\Gamma}=(0,0)^{\top}$, $\mathbf{M}=(0,2\pi/3)^{\top}$, $\mathbf{K}=(4\pi/(3\sqrt{3}),0)^{\top}$, and $\mathbf{K'}=-\mathbf{K}$.  

Neglecting the nonlinear terms and inserting the plane-wave ansatz,
\[\begin{pmatrix}
    x_{n,m}^A \\
    x_{n,m}^B
\end{pmatrix}=\begin{pmatrix}
    a \\
    b
\end{pmatrix}e^{i(n\mathbf{a}_1\cdot \mathbf{k}+m\mathbf{a}_2\cdot \mathbf{k}-\omega t)} + c.c.,\]
in \eqref{eq1}, where $(a,b)^{\top}\in\mathbb{C}^2\backslash\{0\}$ and $c.c.$ denotes complex conjugation, leads to a $2\times 2$ self-adjoint eigenvalue problem that is continuously indexed by the wavevector $\mathbf{k}$.  The resulting real-valued dispersion relation or band structure of the linearization of \eqref{eq1} is given by
\begin{equation}
\label{eq2}
\omega_{\pm}^2(\mathbf{k};\lambda)=\omega_0^2\pm\sqrt{1+2\lambda^2+2\lambda^2\cos{\left(\sqrt{3}k_x\right)}+4\lambda\cos{\left(\dfrac{\sqrt{3}}{2}k_x\right)}\cos{\left(\dfrac{3}{2}k_y\right)}}.
\end{equation}

We refer to the upper (lower) dispersion band as the optical (acoustic) band.  When the intra-cell and inter-cell coupling strengths are equal (i.e. $\lambda=1$), the eigenvalue problem becomes degenerate and the two bands touch at $\mathbf{K}$ and $\mathbf{K'}$ due to the increased discrete rotational symmetry of the lattice's point group.  In a neighborhood of these high-symmetry points, the dispersion relation is approximately linear (see Figure \ref{fig1}(c)) and forms Dirac cones as in the electronic model of graphene~\cite{feff14}.  When the parameter $\lambda$ is tuned away from 1, a bandgap opens at both $\mathbf{K}$ and $\mathbf{K'}$. 

Here, we are interested in the formation of a \textit{semi-Dirac point} in the dispersion relation. Inspecting \eqref{eq2}, we find that the two bands also cross when $\lambda=\lambda_*:=1/2$ at the point $\mathbf{M}$ in the Brillouin zone. Furthermore, at $\lambda_*$ and in a neighborhood of $\mathbf{M}$, we calculate along the $k_y$-axis
\begin{align*} &\omega_{\pm}^2\left(k_x=0,k_y;\lambda_*\right)-\omega_0^2=\pm\dfrac{3}{2}\left|k_y-\dfrac{2\pi}{3}\right|+o\left(\left|k_y-\dfrac{2\pi}{3}\right|\right)\quad\text{as $k_y\to 2\pi/3$.}
\end{align*}
 In the perpendicular direction, we calculate
\begin{align*} &\omega_{\pm}^2\left(k_x,k_y=\dfrac{2\pi}{3};\lambda_*\right)-\omega_0^2=\pm\dfrac{3}{4}k_x^2+o\left(k_x^2\right)\quad\text{as $k_x\to 0$}.  
\end{align*}

The behavior of the band structure at the semi-Dirac point is illustrated in Figure \ref{fig2}.  For $\lambda<\lambda_*$, a direct bandgap opens at $\mathbf{M}$ with width $\omega_+^2(\mathbf{M};\lambda)-\omega_-^2(\mathbf{M};\lambda)=2(1-2\lambda)$ (see the bottom panels in Figure \ref{fig2}). For $\lambda_*<\lambda\leq 1$, the bandgap remains closed (not shown). 
 \begin{figure}[h!]
  \centering
  \includegraphics[scale=0.7]{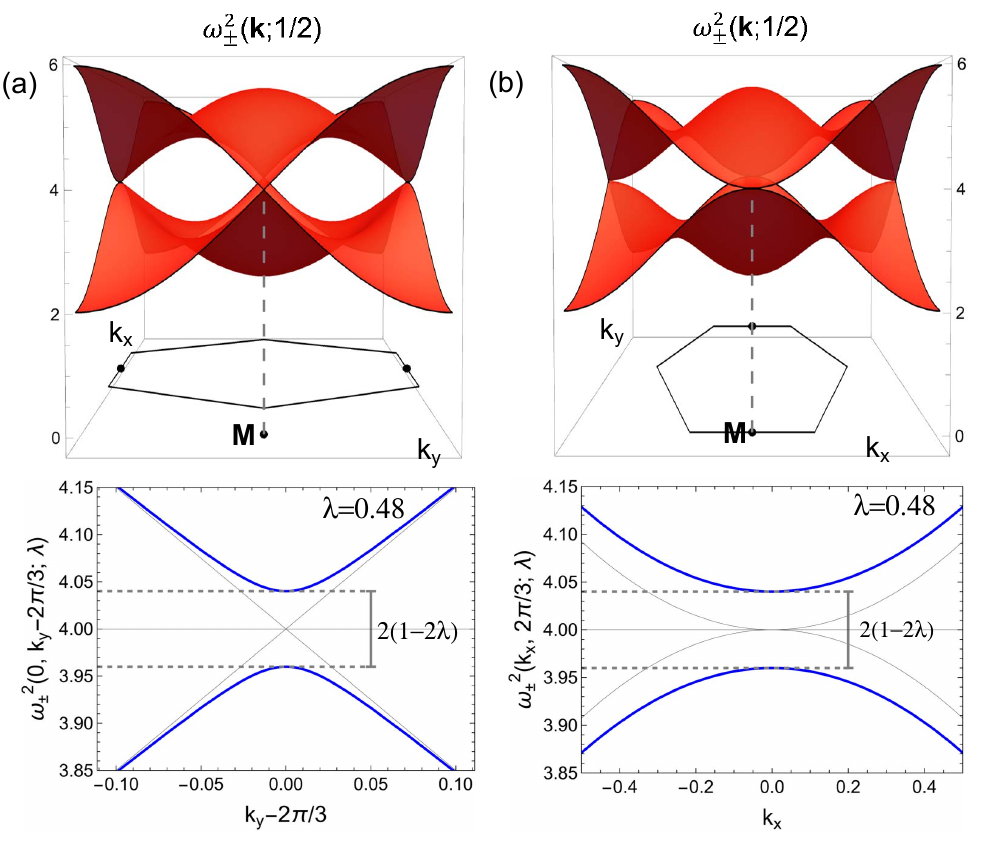}
  \caption{\underline{Top panels}: dispersion relation, \eqref{eq2}, plotted over $\mathcal{B}$ when $\lambda=\lambda_*$, showing the semi-Dirac point along the (a) $k_y-$ and (b) $k_x-$directions. \underline{Bottom panels}: zoomed-in view of the dispersion relation, \eqref{eq2}, centered at $\mathbf{M},$ with an open bandgap (blue curves) for a value $\lambda<\lambda_*$ again along the (a) $k_y-$ and (b) $k_x-$directions.}
  \label{fig2}
\end{figure}

\subsection{Acknowledgements} 
This work was supported by the Air Force Office of Scientific Research (AFOSR) Grant No. FA9550-23-1-0144.
\newpage
\section{Continuum regime near a semi-Dirac point}\label{sec:pde}
\subsection{PDE approximation} To study the behavior of weakly nonlinear wave packets in \eqref{eq1}, centered on the wavevector $\mathbf{M}$ and having a frequency located inside the phonon bandgap (see Figure \ref{fig2}), we define the parameter $\epsilon:=1-2\lambda$.  We consider the regime $0<\epsilon\ll 1$, which corresponds to $\lambda$ taking a value just below $\lambda_*$.

We begin by rewriting system \eqref{eq1} in terms of $\epsilon$:
\begin{align*}
&\ddot{x}_{n,m}^A=-\partial_{z_1}V(x_{n,m}^A,x_{n,m}^B)+\dfrac{1}{2}\left[x_{n,m-1}^B+x_{n+1,m-1}^B\right] 
-\dfrac{\epsilon}{2}\left[x_{n,m-1}^B+x_{n+1,m-1}^B\right]\\ \nonumber
&\ddot{x}_{n,m}^B=-\partial_{z_2}V(x_{n,m}^A,x_{n,m}^B)+\dfrac{1}{2}\left[x_{n,m+1}^A+x_{n-1,m+1}^A\right]-\dfrac{\epsilon}{2}\left[x_{n,m+1}^A+x_{n-1,m+1}^A\right].
\end{align*}

We introduce the following scaled variables to center the analysis at $\mathbf{M}$ and balance the cubic nonlinearity in \eqref{eq1}, 
\begin{align*}
\begin{pmatrix} x_{n,m}^A \\
x_{n,m}^B\end{pmatrix}=\sqrt{\epsilon}\exp\left(i\left(n\mathbf{a}_1+m\mathbf{a}_2\right)\cdot\mathbf{M}\right)\begin{pmatrix} Y_{n,m}^A \\
Y_{n,m}^B\end{pmatrix}
=\sqrt{\epsilon}\exp\left(i\pi m\right)\begin{pmatrix} Y_{n,m}^A \\
Y_{n,m}^B\end{pmatrix}.
\end{align*}

Upon making these substitutions, the system \eqref{eq1} becomes
\begin{align}\label{capy}
\ddot{Y}_{n,m}^A=-\omega_0^2Y_{n,m}^A+&\left[Y_{n,m}^B-\dfrac{1}{2}Y_{n,m-1}^B-\dfrac{1}{2}Y_{n+1,m-1}^B\right]\\ \nonumber 
&+\dfrac{\epsilon}{2}\left[Y_{n,m-1}^B+Y_{n+1,m-1}^B\right]-\epsilon g\left[\left(Y_{n,m}^A\right)^2+\left(Y_{n,m}^B\right)^2\right]Y_{n,m}^A\\ \nonumber
\ddot{Y}_{n,m}^B=-\omega_0^2Y_{n,m}^B+&\left[Y_{n,m}^A-\dfrac{1}{2}Y_{n,m+1}^A-\dfrac{1}{2}Y_{n-1,m+1}^A\right]\\ \nonumber
&+\dfrac{\epsilon}{2}\left[Y_{n,m+1}^A+Y_{n-1,m+1}^A\right]-\epsilon g\left[\left(Y_{n,m}^A\right)^2+\left(Y_{n,m}^B\right)^2\right]Y_{n,m}^B,
\end{align}
where we have used the potential form given in \eqref{potential}.

We define the Fourier representation on the lattice,
\begin{equation}
\begin{pmatrix}
    \hat{Y}_{\mathbf{k}}^A \\
    \hat{Y}_{\mathbf{k}}^B
\end{pmatrix}=
\sum_{(n,m)\in\mathbb{Z}^2}\begin{pmatrix}
    Y_{n,m}^A \\
    Y_{n,m}^B
\end{pmatrix}e^{-i(n\mathbf{a}_1+m\mathbf{a}_2)\cdot \mathbf{k}},
\end{equation}
and its inverse,
\begin{equation}\label{finv}
\begin{pmatrix}
    Y_{n,m}^A \\
    Y_{n,m}^B
\end{pmatrix}
=
\dfrac{1}{\mathcal{A}}\int_{\mathbf{k}\in\mathcal{B}}\begin{pmatrix}
    \hat{Y}_{\mathbf{k}}^A \\
    \hat{Y}_{\mathbf{k}}^B
\end{pmatrix}e^{i(n\mathbf{a}_1+m\mathbf{a}_2)\cdot \mathbf{k}}d\mathbf{k},
\end{equation}
where $\mathcal{A}$ is the area of the Brillouin zone. 

Since we are interested in wavepackets concentrated at $\mathbf{M}$ with a spectral width of the order $\epsilon$, we seek solutions of the form 
\begin{equation}\label{hhat}
\hat{Y}_{\mathbf{k},\epsilon}^{A,B}:=\dfrac{1}{\epsilon^{3/2}}\chi^{A,B}\left(\dfrac{k_x}{\sqrt{\epsilon}},\dfrac{k_y}{\epsilon}\right),
\end{equation}
where $\chi^{A,B}(\mathbf{k})$ are smooth functions that rapidly decay away from the origin so that the integrals in \eqref{four} are convergent. The asymmetric dependence on $\epsilon$ in $k_x$ and $k_y$ is chosen to match the linear-parabolic curvature of the lattice's band structure \eqref{eq2} near $\mathbf{M}.$

We obtain the following approximate expressions for the out-of-cell terms in \eqref{capy} using \eqref{finv} and \eqref{hhat}, performing a Taylor series expansion in $\epsilon$ about the origin and keeping only the leading-order terms.
\begin{align}\label{four}
&Y_{n,m-1}^B+Y_{n+1,m-1}^B\approx\dfrac{1}{\mathcal{A}}\int\limits_{\mathbf{k}\in\mathbb{R}^2}
\left[2-3ik_y\epsilon-\dfrac{3}{4}k_x^2\epsilon\right]e^{i\left(\sqrt{3}(2n+m)\sqrt{\epsilon}k_x+3m\epsilon  k_y\right)/2}\chi^B(\mathbf{k})d\mathbf{k}\\ \nonumber
&Y_{n,m+1}^A+Y_{n-1,m+1}^A\approx\dfrac{1}{\mathcal{A}}\int\limits_{\mathbf{k}\in\mathbb{R}^2}
\left[2+3ik_y\epsilon-\dfrac{3}{4}k_x^2\epsilon\right]e^{i\left(\sqrt{3}(2n+m)\sqrt{\epsilon}k_x+3m\epsilon  k_y\right)/2}\chi^A(\mathbf{k})d\mathbf{k}.
\end{align}
Using \eqref{four}, we introduce the continuum approximation
\begin{align}\label{app}
&Y_{n,m-1}^B+Y_{n+1,m-1}^B\approx 2u_B(t,z,\eta)-2\partial_zu_B(t,z,\eta)+2\partial_{\eta}^2u_B(t,z,\eta)\\ \nonumber
&Y_{n,m+1}^A+Y_{n-1,m+1}^A\approx 2u_A(t,z,\eta)+2\partial_zu_A(t,z,\eta)+2\partial_{\eta}^2u_A(t,z,\eta)
\end{align}
evaluated at $z=m$ and $\eta=\sqrt{2}(2n+m),$ where the functions $u_{A,B}$ are inverse Fourier transforms of the functions $\chi^{A,B}$.  The preceding discussion also motivates the long spatial scalings $Z=\epsilon z$ and $H=\sqrt{\epsilon}\eta.$

Plugging \eqref{app} into \eqref{capy} leads to the PDE approximation of the lattice dynamics, centered on the semi-Dirac point, near the continuum limit 
\begin{align}\label{pde}
&\partial_t^2u_A=-\omega_0^2u_A+\partial_zu_B-\partial_{\eta}^2u_B+\epsilon\left(u_B-\partial_zu_B+\partial_{\eta}^2u_B\right)-\epsilon g\left(u_A^2+u_B^2\right)u_A\\ \nonumber
&\partial_t^2u_B=-\omega_0^2u_B-\partial_zu_A-\partial_{\eta}^2u_A+\epsilon\left(u_A+\partial_zu_A+\partial_{\eta}^2u_A\right)-\epsilon g\left(u_A^2+u_B^2\right)u_B.
\end{align}
\subsection{Multiple-scale analysis}
We define the functions of the multiple-scale independent variables
\[u_{A,B}(t,z,\eta,T,Z,H),\quad\text{where}\quad T=\epsilon t, Z=\epsilon z, H=\sqrt{\epsilon}\eta.\]
Furthermore, we consider the following expansion in powers of $\epsilon$
\[u_{A,B}(t,z,\eta,T,Z,H)=u_{A,B}^{(0)}(t,z,\eta,T,Z,H)+\epsilon u_{A,B}^{(1)}(t,z,\eta,T,Z,H)+\cdots\]

Substituting the above expressions into \eqref{pde} and collecting like-terms in powers of $\epsilon$ gives at zeroth-order:
\[\mathcal{L}_0\begin{pmatrix}
    u_A^{(0)}\\
    u_B^{(0)}
\end{pmatrix}=0,\]
where 
\[\mathcal{L}_0:=\begin{pmatrix}
    \partial_t^2+\omega_0^2 & -\partial_z+\partial_{\eta}^2 \\
    \partial_z+\partial_{\eta}^2 & \partial_t^2+\omega_0^2
\end{pmatrix}.\]
We consider time-harmonic solutions of the form
\begin{equation}
\begin{pmatrix}\label{har}
    u_A^{(0)}\\
    u_B^{(0)}
\end{pmatrix}=\begin{pmatrix}
    U(T,Z,H)\\
    V(T,Z,H)
\end{pmatrix}e^{i\omega_0 t}+c.c.
\end{equation}
where $U$ and $V$ are envelope functions depending only on the long and slow spatiotemporal scale variables.

At first-order, we obtain
\begin{align*}
\mathcal{L}_0\begin{pmatrix}
    u_A^{(1)}\\
    u_B^{(1)}
\end{pmatrix}=\begin{pmatrix}
-2\partial_t\partial_Tu_A^{(0)}+u_B^{(0)}+\partial_Zu_B^{(0)}-\partial_H^2u_B^{(0)}-g\left[\left(u_A^{(0)}\right)^2+\left(u_B^{(0)}\right)^2\right]u_A^{(0)}\\
-2\partial_t\partial_Tu_B^{(0)}+u_A^{(0)}-\partial_Zu_A^{(0)}-\partial_H^2u_A^{(0)}-g\left[\left(u_A^{(0)}\right)^2+\left(u_B^{(0)}\right)^2\right]u_B^{(0)}
\end{pmatrix}.
\end{align*}
Note, since \eqref{har} does not depend on $\eta$, the equation at order $\sqrt{\epsilon}$ in the expansion,
\[ -2\partial_H\partial_\eta u_{A,B}^{(0)}=0\]
is automatically satisfied.

Since we seek bounded solutions, plugging \eqref{har} into the right-hand side of the first-order system, we remove the \textit{secular terms} on the right-hand side by setting the terms with a short time dependence given by $\exp(i\omega_0 t)$ to zero. This leads to the following PDE system
\begin{align}\label{final}
&-2i\omega_0\partial_TU+V+\partial_ZV-\partial_H^2V-2g\left(|U|^{2}+|V|^{2}\right)U=0\\ \nonumber
&-2i\omega_0\partial_TV+U-\partial_ZU-\partial_H^2U-2g\left(|U|^{2}+|V|^{2}\right)V=0.
\end{align}
We refer to \eqref{final} as the semi-Dirac equation.  We note that a similar equation was derived in \cite{dohnal} to model the propagation of optical solitons through planar structures with periodic gratings. However, unlike the model in \cite{dohnal}, equation \eqref{final} has a true bandgap (see \eqref{sddr}).

\subsection{Stationary solutions}
The stationary form of the semi-Dirac equation is found by plugging in the ansatz,
\[\left(U, V\right)^{\top}=\left(U_{}(Z,H;\nu), V(Z,H;\nu)\right)^{\top}\exp\left(i\nu T/(2\omega_0)\right)=:\mathbf{\Psi}\exp\left(i\nu T/(2\omega_0)\right),\]
into \eqref{final}.  

Neglecting the nonlinear terms in the resulting equation, we obtain the self-adjoint spectral problem in $L^2(\mathbb{R}^2,\mathbb{C}^2)$, 
\[\mathcal{L}_{SD}\mathbf{\Psi}:=\left[i\sigma_2\partial_Z+\sigma_1\left(1-\partial_H^2\right)\right]\mathbf{\Psi}=-\nu\mathbf{\Psi},\]
where $\sigma_{j}$ $(j=1,2,3)$ are Pauli matrices and $\sigma_0$ is the identity matrix.  The domain of the operator $\mathcal{L}_{SD}$ is the anisotropic Sobolev space is $\mathcal{D}(\mathcal{L}_{SD}):=\text{H}_Z^1\text{H}_H^2\left(\mathbb{R}^2,\mathbb{C}^2\right)$.  The spectrum of $\mathcal{L}_{SD}$ is gapped and purely essential, given by
\[ \sigma(\mathcal{L}_{SD})=(-\infty, -1]\cup[1,\infty),\]
which can be shown by constructing a Weyl singular sequence of approximate eigenfunctions in $\mathcal{D}(\mathcal{L}_{SD})$~\cite{sigal}.  Note also that since $\mathcal{L}_{SD}$ anticommutes with $\sigma_3$ (it depends only on $\sigma_1$ and $\sigma_2$), the spectrum is symmetric about zero.
\\
\\
\underline{Remark}: Due to the form of nonlinearity, if $\mathbf{\Psi}$ is a stationary solution to \eqref{final} for $g=1$ at $\nu$, then $\sigma_3\mathbf{\Psi}$ is a stationary solution to \eqref{final} for $g=-1$ at $-\nu$. 
\\

Localized solutions to the stationary semi-Dirac equation correspond to the critical points of the functional 
\begin{equation}\label{functional}
    I_{\nu}:=\int_{\mathbb{R}^2}\left(\big\langle\mathcal{L}_{SD}\mathbf{\Psi},\mathbf{\Psi}\big\rangle +\nu |\mathbf{\Psi}|^2-g|\mathbf{\Psi}|^{4}\right)d\mathbf{x}
\end{equation}
for $\mathbf{\Psi}\in\mathcal{D}(\mathcal{L}_{SD}).$  

Like the ``energy" functional for the nonlinear \textit{Dirac equation}, where $\mathcal{L}_{SD}$ in \eqref{functional} is replaced by the operator
\begin{equation}\label{dirac}
\mathcal{L}_D\mathbf{\Psi}:=\left[-i\sigma_1\partial_Z-i\sigma_2\partial_H+\sigma_3\right]\mathbf{\Psi},
\end{equation}
$I_{\nu}$ is strongly indefinite, making variational arguments about the existence of critical points more difficult than those for nonlinear Sch\"{o}dinger (NLS) equations, where the spectrum is bounded from below~\cite{sere,IW10}. 

For the \textit{1D Dirac equation}, explicit localized stationary solutions (gap solitons) have been obtained in the literature for various nonlinearities~\cite{pelinovsky1, sax}. For instance, the massive Thirring~\cite{thirring} model is used to describe self-interacting fermions in one spatial dimension with a Lorentz-invariant form of nonlinearity and it is known to be completely integrable.  Another example is the explicit gap solitons found in \cite{aceves} to a 1D Dirac equation with a cubic Kerr nonlinearity used to model laser propagation through optical fibers with periodic gratings.

In the following theorem, we find explicit line-soliton solutions to \eqref{final} using known solutions to a related equation and exploiting the unitary symmetry of both systems.  Within the bandgap, the stationary solutions obtained are smooth and exponentially localized in the $Z$ direction, taking a particularly simple form at the center of the bandgap ($\nu=0$).  The decay rate of the solution in the $Z$ direction becomes algebraic in the limit approaching one edge of the gap and the solution is identically zero in the limit at the other edge. In Section \ref{sec:anti} we compare the derived solitons with numerically computed discrete breathers near the semi-Dirac point. 
\begin{theorem}\label{gapsol}
The semi-Dirac equation, \eqref{final}, has a family of explicit stationary solutions, $\mathbf{\Phi}(Z;\nu)=\left(\Phi_1,\Phi_2\right)^{\top}$, which are localized in the $Z$ direction and constant in the $H$ direction (i.e., line-solitons) for each $\nu\in(-1,1)$.  For $g=1$, we have
\begin{enumerate}
    \item $\lim_{\nu\to +1}\mathbf{\Phi}(Z;\nu)=\dfrac{1+i}{\sqrt{2}(1+4Z^2)}\left(1-2Z,1+2Z\right)^{\top}$
    \item $\mathbf{\Phi}(Z;0)=\dfrac{1+i}{2\sqrt{2}}\textnormal{sech}\left(2Z\right)\left(e^{-Z},e^{Z}\right)^{\top}$
    \item $\lim_{\nu\to-1}\mathbf{\Phi}(Z;\nu)=(0,0)^{\top}.$
\end{enumerate}
\end{theorem}

\noindent\underline{Remark}: By the preceding remark, the above theorem holds when $g=-1$ by replacing $\nu$ with $-\nu.$
\begin{proof}
We introduce the related 1D system,
\begin{equation}\label{pel}
  \left(\sigma_1+i\sigma_3\partial_Z+\nu\sigma_0\right)\mathbf{\Psi}=2g|\mathbf{\Psi}|^2\mathbf{\Psi},  
\end{equation}
like the one considered in \cite{pelinovsky2}, for $\mathbf{\Psi}=\left(U(Z), V(Z)\right)^{\top}.$  Making the ansatz, $\mathbf{\Psi}_{\nu}=\left(U_{\nu}(Z), U_{\nu}^*(Z)\right)^{\top},$ where $z^*$ denotes the complex conjugate of $z\in\mathbb{C}$, \eqref{pel} becomes the scalar equation
\[i\dfrac{dU_{\nu}}{dZ}+\nu U_{\nu}+U_{\nu}^*-4g|U_{\nu}|^2U_{\nu}=0,\]
which has the explicit smooth and exponentially localized solution family
\begin{equation}\label{sol}
    U_{\nu}(Z)=\dfrac{\sqrt{1-\nu^2}\left[\sqrt{1-\nu}\text{cosh}\left(\sqrt{1-\nu^2}Z\right)-i\sqrt{1+\nu}\text{sinh}\left(\sqrt{1-\nu^2}Z\right)\right]}{\sqrt{2}\left[\text{cosh}\left(2\sqrt{1-\nu^2}Z\right)-\nu\right]}
\end{equation}
for $\nu\in(-1,1)$~\cite{pelinovsky2}.

If $\mathbf{\Psi}$ is a solution to \eqref{pel}, then $\mathbf{\tilde{\Psi}}:=\mathcal{U}\mathbf{\Psi},$ for $\mathcal{U}\in\text{SU}(2),$ is a solution to 
\[\mathcal{U}\left(\sigma_1+i\sigma_3\partial_Z+\nu\sigma_0\right)\mathcal{U}^{\dagger}\mathbf{\tilde{\Psi}}=2g|\mathbf{\tilde{\Psi}}|^2\mathbf{\tilde{\Psi}},\]
where $\mathcal{U}^{\dagger}$ denotes the conjugate transpose of the matrix $\mathcal{U}$. 

Letting $\mathcal{U}=\left(\sigma_0+i\sigma_1\right)/\sqrt{2},$ we obtain the stationary semi-Dirac equation,
\[\left(\mathcal{L}_{SD}+\nu\right)\mathbf{\tilde{\Psi}}=2g|\mathbf{\tilde{\Psi}}|^2\mathbf{\tilde{\Psi}}.\]
Using \eqref{sol}, we get the solution to the semi-Dirac equation
\begin{equation}\label{gapsoleq}
\mathbf{\Phi}(Z;\nu)=\dfrac{1+i}{2}\big(\text{Re}\left(U_{\nu}\right)+\text{Im}\left(U_{\nu}\right), \text{Re}\left(U_{\nu}\right)-\text{Im}\left(U_{\nu}\right)\big)^{\top}.
\end{equation}
 Taking limits gives the result.
\end{proof}

\underline{Remark}: The line-solitons above are infinite energy states in the bandgap since they are de-localized in the $H$ direction.  A formal proof for the existence or non-existence of fully \textit{localized solutions in 2D} to the semi-Dirac equation, \eqref{final}, is still open.  The variational arguments in \cite{sere} do not apply for the nonlinearity in \eqref{final}. It was shown in \cite{borelli} that the 2D \textit{Dirac equation}, with the linear terms in \eqref{dirac} and the same form of nonlinearity considered here, admits smooth and exponentially localized solutions inside the gap using a shooting method argument.  However, the result relies on a radially symmetric solution ansatz, originally introduced in \cite{caz}, to solve the compactness issue of a critical Sobolev embedding in $L^4(\mathbb{R}^2,\mathbb{C}^2)$, which is not useful for the semi-Dirac equation due to the anisotropic appearance of spatial derivatives in $\mathcal{L}_{SD}.$  
\newline

From the 1D stationary solution \eqref{gapsoleq} in Theorem \eqref{gapsol} it can be shown that  
\[\mathbf{\Phi}(Z;\nu)\sim\dfrac{\sqrt{1-\nu^2}}{4}(1+i)\text{sech}\left(\sqrt{1-\nu^2}Z\right)\left[1,1\right]^{\top}\text{ as }\nu\to-1.\]

Indeed, near the edges of the essential spectrum, the NLS equation is the leading-order approximation to the semi-Dirac equation. Heuristically, near $\nu=-1$, if we ignore the other branch of the spectrum in the semi-Dirac equation, after a simple re-scaling we obtain the following \textit{focusing} NLS equation 
\[\left(\mathcal{L}_{SD}+\nu-\left|\mathbf{\Psi}\right|^2\right)\mathbf{\Psi}\approx\left(1+\nu-\Delta-\left|\mathbf{\Psi}\right|^2\right)\mathbf{\Psi},\]
which is obtained using the Taylor expansion of the negative branch of the semi-Dirac equation's dispersion relation,
\begin{equation}\label{sddr}
   \nu_{\pm}(\mathbf{k})=\pm\sqrt{(1+k_H^2)^2+k_Z^2}, 
\end{equation}
near the origin.  Making the same approximation near $\nu=1$ gives the \textit{de-focusing} NLS equation which is known to not admit localized solutions~\cite{fibich}. This makes it clear that both branches of the essential spectrum shape gap solitons, particularly for $\nu$ away from $-1$.

Looking for localized solutions in both $Z$ and $H$, we make the following separable solution ansatz for the semi-Dirac equation
\begin{equation}\label{anz}
\begin{pmatrix}
    U(Z,H;\nu) \\
    V(Z,H;\nu)
    \end{pmatrix}
=\begin{pmatrix}
\Phi_1(Z;\nu)G_{1}(H;\nu) \\
\Phi_2(Z;\nu)G_{2}(H;\nu)
    \end{pmatrix},
\end{equation}
where $\mathbf{
\Phi}(Z;\nu)$ is the localized 1D solution from Theorem \ref{gapsol}.  

Plugging into the stationary form of \eqref{final} gives
\begin{align*}
& F_1(\mathbf{\Phi}):=\dfrac{2g\left(|\Phi_1|^2+|\Phi_2|^2\right)\Phi_1}{\Phi_2}=\dfrac{G_{1}''}{\left(1-|G_{1}|^2\right)G_{1}}\\
&F_2(\mathbf{\Phi}):=\dfrac{2g\left(|\Phi_1|^2+|\Phi_2|^2\right)\Phi_2}{\Phi_1}=\dfrac{G_{2}''}{\left(1-|G_{2}|^2\right)G_{2}}.
\end{align*}
We note that $F_1$ and $F_2$ are real-valued functions.  Using the expression for the gap soliton in item $2.$ of Theorem \ref{gapsol} at the middle of the gap ($\nu=0$) and for $g=1$, we simplify to obtain
\[F_1\left(\mathbf{\Phi}(Z;0)\right)=1-\tanh(2Z)\quad\text{ and }\quad F_2\left(\mathbf{\Phi}(Z;0)\right)=1+\tanh(2Z).\]

We then have the limiting behaviors as $Z\to\pm\infty$,
\begin{align*}
&F_{1,-}:=\lim_{Z\to-\infty}F_1\left(\mathbf{\Phi}(Z;0)\right)=2\implies -2G_{1,-}+G''_{1,-}+2|G_{1,-}|^2G_{1,-}=0\\
&F_{2,-}:=\lim_{Z\to-\infty}F_2\left(\mathbf{\Phi}(Z;0)\right)=0\implies G''_{2,-}=0
\end{align*}
and
\begin{align*}
&F_{1,+}:=\lim_{Z\to+\infty}F_1\left(\mathbf{\Phi}(Z;0)\right)=0\implies G''_{1,+}=0 \\
&F_{2,+}:=\lim_{Z\to+\infty}F_2\left(\mathbf{\Phi}(Z;0)\right)=2\implies-2G_{2,+}+G''_{2,+}+2|G_{2,+}|^2G_{2,+}=0,
\end{align*}
where the $\pm$ subscripts above denote the functions at these limits.
It is clear that $F_{1,2}\left(\mathbf{\Phi}(Z;0)\right)$ approach these constant values at an exponential rate. Indeed, we have $|F_{j,\pm}-F_j|=\text{O}\left(\exp(\mp 4Z)\right)$ for $j=1,2$ as $Z\to\pm\infty.$ The nontrivial limiting equation for $G_{j,\pm}$ above is the scalar focusing NLS equation. 

Thus, although we do not have a global localized solution to the stationary form of \eqref{final}, we have the two separable bounded solutions as $Z\to\pm\infty$ given, respectively, by
\begin{equation}\label{sep}
\mathbf{\Psi}_{-}:=\begin{pmatrix}\Phi_1(Z;0)G_{1,-}(H;0)\\
0
\end{pmatrix}\quad\text{ and }\quad\mathbf{\Psi}_+:=\begin{pmatrix}0\\
\Phi_2(Z;0)G_{2,+}(H;0)
\end{pmatrix},
\end{equation}
 where
\[G_{1,-}(H;0)=G_{2,+}(H;0)=\sqrt{2}\text{sech}\left(\sqrt{2}H\right).\]
We note that the separation in \eqref{anz} only gives a nontrivial result in the middle of the gap. Indeed, we have $F_{1,2}\left(\mathbf{\Phi}(Z;\nu)\right)\to 0$ as $Z\to\pm\infty$ for nonzero $\nu\in(-1,1).$

In what follows, we numerically compute mid-gap discrete breathers on the lattice near the continuum limit at which the effective semi-Dirac equation approximation is valid.  We find that the spatial tails of these breathers are accurately captured by these exact separable profiles for sufficiently large $Z$. 

\section{Anti-continuum regime and discrete breathers}\label{sec:anti}
\subsection{Existence of discrete breathers}

We rewrite equation \eqref{eq1} in the form
\begin{equation}\label{disc}
F(X^{\lambda}(t);\lambda):=\begin{pmatrix} 
\ddot{x}_{n,m}^A+\partial_{z_1}V(x_{n,m}^A,x_{n,m}^B)-\lambda\left[x_{n,m-1}^B+x_{n+1,m-1}^B\right]\\ 
\ddot{x}_{n,m}^B+\partial_{z_2}V(x_{n,m}^A,x_{n,m}^B)-\lambda\left[x_{n,m+1}^A+x_{n-1,m+1}^A\right]\end{pmatrix}=\begin{pmatrix} 0 \\
0
\end{pmatrix}.
\end{equation}
We restrict the domain of the mapping $F(X^{\lambda};\lambda)$ to square-summable sequences, that is $X^{\lambda}(t)\in \ell^2\left(\mathbb{Z}^2,\mathbb{R}^2\right)\cong\ell^2(\mathbb{Z})$.  Additionally, we restrict the domain so that $X^{\lambda}(t)$ is $T_b$--periodic and even in time.  Finally, we restrict the domain so that $X^{\lambda}(t)$ has the required regularity, i.e. each of its elements is in the Sobolev space $\text{H}^2(\mathbb{R}).$  We denote the resulting domain of $F(X^{\lambda};\lambda)$ by $\mathcal{D}\times\mathbb{R},$ where
\[\mathcal{D}:=\left\{X(t)=\{X(t)\}_{n\in\mathbb{Z}}\in\ell^2(\mathbb{Z})\quad |\quad X(t)\in \text{H}^2(\mathbb{R}/\mathbb{Z}T_b)\text{ and }X(t)=X(-t)\right\}.\]

We then seek a nontrivial solution to \eqref{disc} in the simplified \textit{anti-continuum limit} ($\lambda=0$), $X^{0}\in\mathcal{D},$ for instance, of the form
\begin{equation}\label{anti}
    X_*(t):=X^0(t)=\left\{\cdots,0,0,\begin{pmatrix}x_{*}^A(t) \\
x_{*}^B(t)
    \end{pmatrix},0,0,\cdots\right\},
\end{equation}
with fixed frequency $\omega_b=2\pi/T_b.$  

Using the implicit function theorem~\cite{nirenberg, Hof:23}, it can be shown that $X^0$ smoothly continues in $0\leq\lambda<\lambda_b$, for some nonzero value $\lambda_b$, to $X^{\lambda}\in\mathcal{D}$ if two conditions are satisfied: 
\begin{enumerate}
 \item (Non-resonance) for all $n\in\mathbb{Z}$ we have that 
 \begin{equation}\label{nonres}
 \left(n\omega_b\right)^2\neq \Omega_{\pm}^2:=\omega_0^2\pm 1.
 \end{equation}
 \item (Non-degeneracy) the null-space of the operator
\begin{equation}\label{nondeg}
    L_*:=\sigma_0\dfrac{d^2}{dt^2}+\mathcal{H}_{V_*},
\end{equation}
where $\mathcal{H}_{V_*}$ is the Hessian matrix of the potential $V(z_1,z_2)$ evaluated at $\left(x_*^A, x_*^B\right)^{\top}$, is empty in $\mathcal{D}.$
\end{enumerate}

In the anti-continuum limit \eqref{anti}, there is an \textit{in-phase} and \textit{out-of-phase} one-parameter family of explicit solutions in the respective cases $g=\pm 1$ given by
\begin{equation}\label{cn}
\begin{pmatrix}
    x_*^A(t)\\
    x_*^B(t)
\end{pmatrix}^{\pm}=a_*\text{cn}\left(\sqrt{(\Omega_{\mp}^2\pm 2a_*^2)} t,\dfrac{\pm a_*^2}{(\Omega_{\mp}^2\pm 2a_*^2)} \right)\begin{pmatrix}
    1 \\
    \pm 1
\end{pmatrix},
\end{equation}
where $\text{cn}$ is the even Jacobi elliptic function.  These solutions are derived from the reduction of the anti-continuum limit system, which consists of two coupled second-order ODEs,  into a single equation as shown in \cite{Hof:23} in the in-phase and out-of-phase cases.  The resulting second-order ODE is in the form of the integrable Duffing equation with exact solutions~\cite{verhulst}.      

The periods of the solutions in \eqref{cn} are given by
\begin{equation}\label{per}
    T_b^{\pm}(a_*^2)=\dfrac{4}{\sqrt{\Omega_{\mp}^2\pm 2a_*^2}}\mathcal{K}\left(\sqrt{\dfrac{\pm a_*^2}{(\Omega_{\mp}^2\pm 2a_*^2)}}\right),
\end{equation}
where $\mathcal{K}$ is the complete elliptic integral of the first kind.  Note that periodic solutions exist in the anti-continuum limit, \eqref{cn}, with any frequency $\omega_b>\Omega_-$ in the hardening case and $0<\omega_b<\Omega_+$ in the softening case.   As the solution amplitude, $a_*$, approaches zero, the time-dependent term in \eqref{cn} approaches $a_*\cos(\Omega_{\mp}t)$ and $T_b^{\pm}$ in \eqref{per} approaches $2\pi/\Omega_{\mp}.$  Both the non-resonance and non-degeneracy conditions above fail in this limit.  For almost every possible nonzero $a_*$, both conditions are satisfied and imply that discrete breathers with almost any $\omega_b$ in the ranges above exist on the globally coupled lattice in \eqref{eq1}, at least up to some coupling strength $\lambda_b>0$ (see \cite{Hof:23} for details).

We note that the continued breathers are \textit{exponentially localized} on the lattice. This is true for the nearest-neighbor couplings considered here, but also more generally for an infinite number of couplings so long as at each site there exists a positive number $R$ such that the coupling strengths decay exponentially outside a circle of radius $R$ centered on that site.  A formal proof of this claim is given in \cite{Hof:23}. 

\subsection{Numerical continuation of discrete breathers}
Discrete breathers to \eqref{disc} ($\lambda\neq 0$) can be found using a numerical continuation scheme, seeded from the explicit anti-continuum states in \eqref{cn}. Since breathers are $T_b$-periodic, we represent system \eqref{disc} in the Fourier domain using the discrete Fourier transform, $\mathcal{F},$ in time and its inverse, $\mathcal{F}^{-1},$ given by
\[\hat{x}[p]:=\mathcal{F}\{x\}[p]=\sum_{j=0}^{N-1}x[j]e^{-2\pi i j p/N},\quad x[j]=\mathcal{F}^{-1}\{\hat{x}\}[j]=\dfrac{1}{N}\sum_{p=0}^{N-1}\hat{x}[p]e^{2\pi i j p/N},\]
where $N$ is the number of equally-spaced sampled points of $x(t)$ in the interval $[0, T_b].$  Furthermore, since we seek even-in-time solutions, only cosine terms are needed and the system size is reduced by half. 

The left-hand side of equation \eqref{disc} then becomes
\begin{equation}
\mathcal{F}\left\{F(X^{\lambda};\lambda)\right\}=\begin{pmatrix} 
-p^2\omega_b^2\hat{x}_{n,m}^A+\mathcal{F}\left\{\partial_{z_1}V(x_{n,m}^A,x_{n,m}^B)\right\}-\lambda\left[\hat{x}_{n,m-1}^B+\hat{x}_{n+1,m-1}^B\right]\\ 
-p^2\omega_b^2\hat{x}_{n,m}^B+\mathcal{F}\left\{\partial_{z_2}V(x_{n,m}^A,x_{n,m}^B)\right\}-\lambda\left[\hat{x}_{n,m+1}^A+\hat{x}_{n-1,m+1}^A\right]\end{pmatrix}
\end{equation}
for $p=0,1,\cdots,N-1.$

In the Fourier domain and for $\lambda=0$, the linearized mapping of \eqref{disc} about $X_*$ is represented by the block diagonal matrix 
\begin{equation}\label{matrix}
\mathcal{J}(X_*;0):=\begin{bmatrix}\ddots & & & \\
 & D_{\omega_b} & \text{-}I & \\
 & \text{-}I & D_{\omega_b} & \\
& & &\ddots \end{bmatrix}+
\begin{bmatrix}\ddots & & & & & \\
& 0 & & & &\\
&  & J_{11}^* & J_{12}^*  & &\\
 & &  J_{21}^* & J_{22}^*  & &\\
 & & & & 0&\\
& & & & &\ddots\end{bmatrix}.
\end{equation}
Above, $I$ is the $N\times N$ identity matrix, $D_{\omega_b}$ is the $N\times N$ diagonal matrix, \[D_{\omega_b}:=\text{diag}\left(\omega_0^2-\left(0\omega_b\right)^2,\omega_0^2-\left(1\omega_b\right)^2,\cdots,\omega_0^2-\left((N-1)\omega_b\right)^2\right),\] 
and $J_{\alpha\beta}^*$ ($\alpha,\beta=1,2$) are $N\times N$ \textit{circulent matrices} generated from the vectors
\[\mathbf{v}_{\alpha\beta}^*=N^{-1}\mathcal{F}\left\{\partial_{z_{\alpha}}\partial_{z_{\beta}}V(x_*^A,x_*^B)-\partial_{z_{\alpha}}\partial_{z_{\beta}}V(0,0)\right\}[p],\quad p=0,1,\cdots, N-1.\]
The circulent matrices arise due to convolution in the Fourier domain. Note that when the non-resonance condition \eqref{nonres} is violated, an infinite number of rows in \eqref{matrix} are scalar multiples of one another, and $\mathcal{J}$ is singular.

For a \textit{finite} honeycomb lattice consisting of $M$ unit-cells, the matrix $\mathcal{J}$ is of size $2MN\times2MN.$ See Figure \ref{fig3}(a) for a schematic of the global hexagonal lattice used.  When $\lambda\neq 0,$ $\mathcal{J}(X_*;0)$ in \eqref{matrix} is replaced by the symmetric block matrix $\mathcal{J}(X_*,\lambda)$ and, for the nearest-neighbor couplings considered here, each row in the first matrix on the right-hand side of \eqref{matrix} has \textit{two} additional nonzero blocks, $\lambda I,$ with the exception of those rows corresponding to cells on the boundary of the lattice where we enforce the Dirichlet condition. We now implement a Newton-Raphson scheme in the temporal Fourier domain to solve \eqref{disc} for some $\lambda\neq 0$ with the initial guess $X_*,$ where at each Newton iteration we use the generalized minimal residual method (GMRES) to solve the resulting system within a set tolerance. We note that since the seeded solutions in \eqref{cn} have exponentially decaying Fourier coefficients, only a small number of harmonics are needed to achieve convergence (here we choose either $N=8$, $16$, or $32$).   
\begin{figure}[h!]
  \includegraphics[scale=0.59]{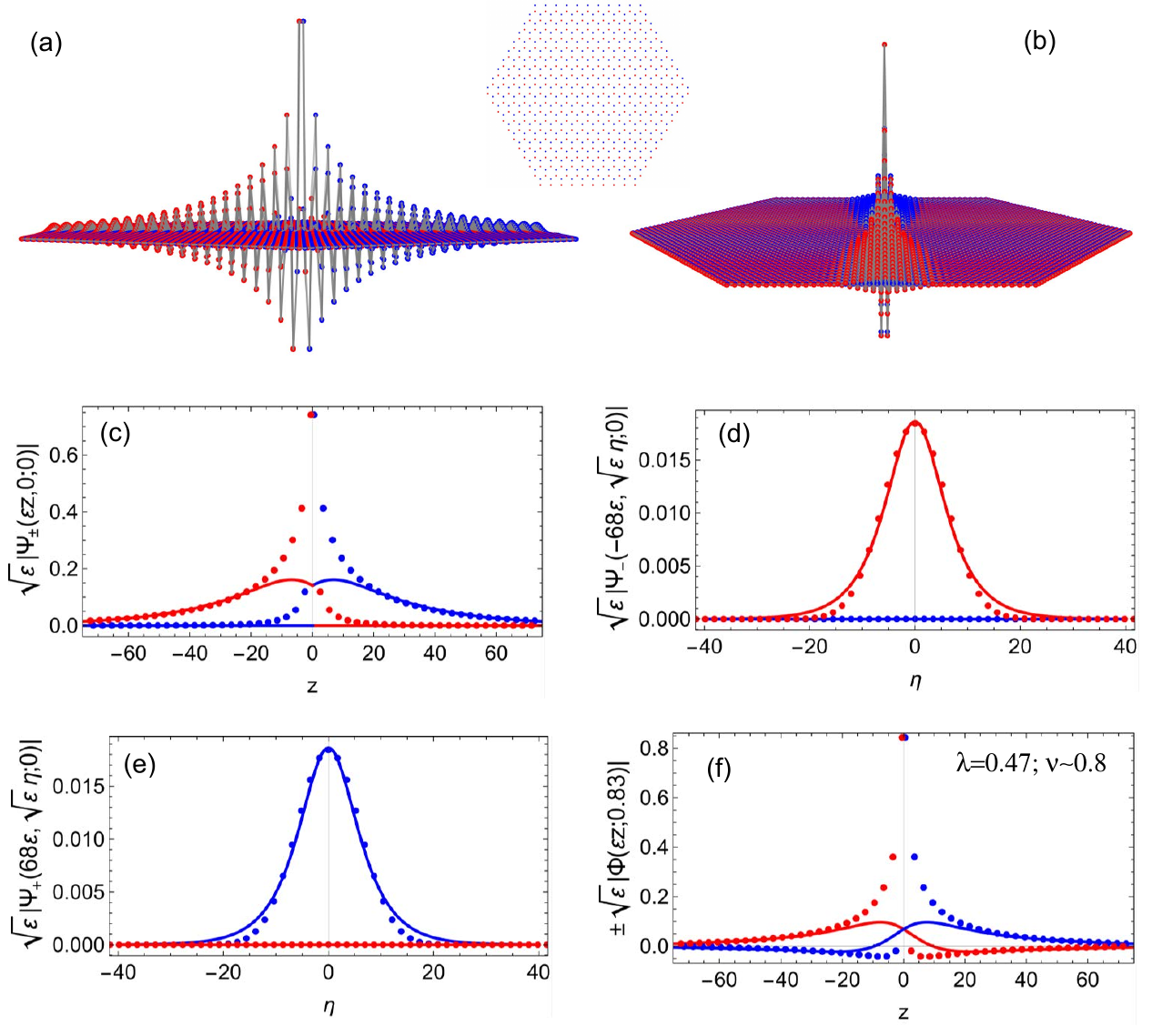}
  \caption{Inset to (a) shows the truncated honeycomb lattice with hexagonal Dirichlet boundaries used in computations; (a) view of the converged spatial profile at $t=0$ of the in-phase midgap ($\omega_b=\omega_0=2$) discrete breather when $\lambda=0.48\lesssim\lambda_*=1/2$, seeded from the $(+)$-state in \eqref{cn} ($g=1$) with $a_*\approx 0.82025$, $N=8$, and $M=7500$; (b) the same breather shown in (a), viewed at a perpendicular angle to the one in (a); (c) ($\pm$)-asymptotic solutions to the continuum theory in \eqref{scaled} along the $z$-axis (solid curves) for the same value of $\lambda$ as in (a) and (b) superimposed with the peak 1D slice of breather in (a) along the same direction; (d) profile of the (-)-asymptotic solution in \eqref{scaled} along the $\eta$-axis at fixed $z\approx -68$ superimposed with the corresponding 1D slice of the breather in (b); (e) same as (d) but for (+)-asymptotic solution in \eqref{scaled} and $z\approx +68$; (f) profile of a 1D gap soliton in Theorem \ref{gapsol} along $z$ at a frequency off-center inside the gap with nonzero $\nu\approx0.83$ and $\lambda=0.47$  superimposed with the numerically computed 2D breather with the same $\lambda$ and corresponding frequency ($\omega_b=\omega_0+\epsilon\nu/(2\omega_0)\approx 2.0125$) seeded from the $(+)$-state in \eqref{cn} with $a_*\approx 0.841261$.}
  \label{fig3}
\end{figure}

Figure \ref{fig3} shows the spatial profiles when $t=0$ of two discrete breathers, one with a frequency located in the center of a small phonon bandgap ($\omega_b=\omega_0$) and the other just slightly above the center of the gap ($\omega_b\gtrsim \omega_0$).  The breathers are seeded in the central unit cell of the global lattice from the $\pm$-states in \eqref{cn} and numerically continued to a value just below $\lambda_*=1/2.$  As $\lambda$ approaches $\lambda_*$, the localization strength of the breather dramatically weakens and its central amplitude starts to decay.  This behavior is observed in all of the panels in Figure \ref{fig3}, where the spatial tails of the breathers are spread over a lattice with $7,500$ unit-cells when $\lambda=0.48.$ The extended tail profiles agree well with the separable stationary solutions to our leading-order PDE approximation for $Z$ sufficiently large in \eqref{sep}. We are unable to numerically resolve this breather past $\lambda=0.48$ or below $\epsilon=0.04$ as it continues to decay and delocalize.  

Figure \ref{fig3} first shows the 2D spatial profile of the in-phase midgap breather at two perpendicular viewing angles.  A schematic of the lattice geometry used in computations is shown in the inset of panel (a). Panel (a) shows the breather's profile along the direction corresponding to the $Z$-axis of the continuum theory from Section \ref{sec:pde} and panel (b) shows the same breather along the transverse direction corresponding to the $H$-axis. The \textit{out-of-phase} midgap breather can also be continued to $\lambda=0.48$ (in this case $g=-1$) and has an envelope profile nearly identical to the in-phase breather and is therefore not shown.   

To compare the midgap breather tails with the asymptotic solutions to the continuum theory in \eqref{sep}, we expand the complex modulus of the long-wave envelopes for $Z\to\pm\infty$ in terms of their original spatial scalings,
\begin{align}\label{scaled}
&\sqrt{\epsilon}|\mathbf{\Psi}_{-}(\epsilon z,\sqrt{\epsilon}\eta;0)|=\sqrt{\dfrac{1-2\lambda}{2}}e^{-(1-2\lambda)z}\text{sech}\left(2(1-2\lambda)z\right)\text{sech}\left(\sqrt{(1-2\lambda)}\eta\right)
\begin{pmatrix}
   1 \\
   0
\end{pmatrix}\\ \nonumber
&\sqrt{\epsilon}|\mathbf{\Psi}_{+}(\epsilon z,\sqrt{\epsilon}\eta;0)|=\sqrt{\dfrac{1-2\lambda}{2}}e^{(1-2\lambda)z}\text{sech}\left(2(1-2\lambda)z\right)\text{sech}\left(\sqrt{(1-2\lambda)}\eta\right)
\begin{pmatrix}
   0 \\
   1
\end{pmatrix}.
\end{align}
Panel (c) in Figure \ref{fig3} shows the scaled $\pm$-profiles in \eqref{scaled} with solid curves along the $z$-axis for $z\geq 0$ ($z\leq 0$), respectively, at $\lambda=0.48$ superimposed with the 1D slice of the central peak of the breather shown in panel (b) along the same direction.  Panel (c) clearly shows that the spatial tails (for sufficiently large $|z|$) of this breather are accurately captured by the exact envelope expressions in \eqref{scaled}.  The envelope approximation for the tails of the out-of-phase breather is identical (not shown). 

In panels (d) and (e) of Figure \ref{fig3}, we plot the envelope profile $\sqrt{\epsilon}|\mathbf{\Psi}_{\mp}|$ in \eqref{scaled} along the $\eta$-axis for a fixed large value of $|z|$ at $\mp|z|$, respectively, both superimposed with the corresponding 1D slice of the breather in (b).  Again, the breather tails are well-approximated by the leading-order asymptotic profiles, with the exception of the very low amplitude tails in the $
\eta$-direction.  A possible reason for this discrepancy is the fact that the asymptotic theory in $\epsilon$ is not uniform in $z$ and $\eta$. Indeed, here $\epsilon=0.04$ and so $\sqrt{\epsilon}=0.2$, an order of magnitude larger. However, since convergence dramatically slows down past $\lambda=0.48$, we are unable to numerically resolve the breather for smaller $\epsilon$ where the agreement should improve. Additionally, the size of the computational domain, $M$, needed to  accurately capture the breather tails for $Z=\epsilon z\gg1$ becomes restrictively large as $\epsilon\downarrow 0$.  

Our numerics suggest that long-lived in- and out-of-phase breathers exist near the semi-Dirac setting, i.e., with $\lambda\lesssim\lambda_*=1/2,$ for other frequencies inside the gap. Indeed, when the frequency of the breather shown in Figure \ref{fig3} (a)-(e) is varied above and below the middle of the gap, but still within the gap, our numerical scheme may still converge. We remark that the tails of the converged breather again closely follow the envelope of the 1D gap soliton of the continuum theory in Theorem \ref{gapsol} along $z$ for the corresponding value of nonzero $-1<\nu<1$. However, as mentioned earlier, the presented asymptotic separation method of the semi-Dirac equation in Section \ref{sec:pde} does not produce a nontrivial result for nonzero $\nu$. 

In panel (f) of Figure \ref{fig3} we show an in-phase breather with an \textit{off-center} frequency located in the upper-half of the small gap, corresponding to a value of $\nu\approx0.83$.  Indeed, at leading-order in $\epsilon$ the gap is 
\[\omega_+(\mathbf{M};\epsilon)-\omega_-(\mathbf{M};\epsilon)=\sqrt{\omega_0^2+\epsilon}-\sqrt{\omega_0^2-\epsilon}\approx\dfrac{\epsilon}{\omega_0}\]
and $\nu\epsilon/(2\omega_0)$ represents the off-set from the center of the small gap for nonzero $-1<\nu<1.$
The breather is superimposed with the corresponding modulus of the gap line soliton, $\sqrt{\epsilon}|\mathbf{\Phi}(\epsilon z;0.83)|$, and again closely follows the breather tails for sufficiently large $|z|.$  

Finally, we remark that the amplitudes of in-phase (out-of-phase) breathers with hardening (softening) nonlinearity do not appear to decay to zero near the edge of the optical (acoustic) band, consistent with Theorem \ref{gapsol}.  Furthermore, in both cases, the exponential decay rate of the breather's tails in the $z$ direction decreases as its frequency nears these edges, approaching an algebraic decay like the gap solitons in Theorem \ref{gapsol}.

\subsection{Dynamical stability of discrete breathers}
The spectral stability of discrete breathers is determined by the linearized system about $X^{\lambda}(t).$ This entails solving a system of $4M$ coupled $T_b$-periodic linear differential equations. 
\begin{figure}[h!]
  \centering
  \includegraphics[scale=0.65]{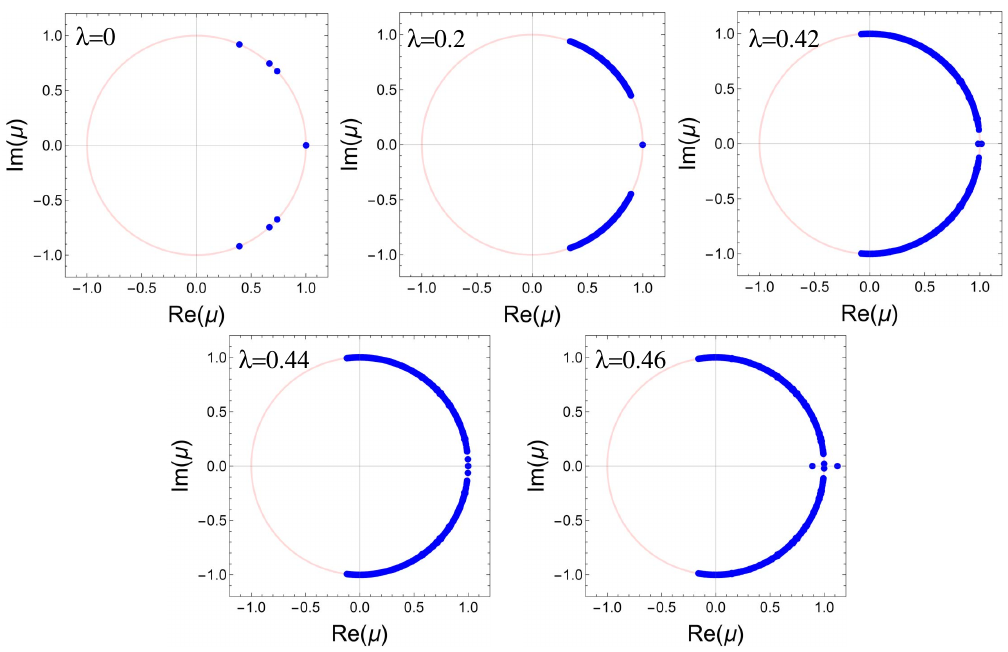}
  \caption{Floquet multipliers, $\mu,$ at different values of $\lambda$ on the same continuation curve as the midgap breather shown in Figure \ref{fig3}. }
  \label{fig4}
\end{figure}
\noindent  From Floquet theory~\cite{hill}, we construct the $4M\times 4M$ \textit{monodromy matrix}, where the $j$th column is the solution vector evaluated at $T_b$ of the system of $4M$ first-order linear equations with initial conditions $(0,\cdots, 1,0,\cdots,0),$ where the $1$ is in the $j$th position. The eigenvalues, $\mu,$ of the monodromy matrix are known as Floquet multipliers.
Since \eqref{disc} is Hamiltonian, it is well-known that the breather is linearly spectrally stable if and only if all of the Floquet multipliers lie along the unit circle in the complex plane; otherwise perturbations to the breather will grow exponentially with time.  

Figure \ref{fig4} shows the Floquet multipliers of the midgap breather seeded from the $(+)$-state in \eqref{cn} as shown in Figure \ref{fig3} at different values of $0\leq\lambda<\lambda_*.$  We remark that the Floquet multipliers of the midgap breather seeded from the $(-)$-state in \eqref{cn} are very similar and so are not shown here.  There are two Floquet multipliers for every $\lambda$ fixed at $+1$ (see \cite{Hof:23} or \cite{ma96} for a discussion). We note that the precise numerical resolution of these multipliers is difficult \cite{flasch08} and in Figure \ref{fig4} they sometimes do not appear to be located exactly at $+1$ although in fact they always are. 

Figure \ref{fig4} shows that when $\lambda=0$, in addition to the two Floquet multipliers at $+1$, there are three other distinct multipliers (along with their complex conjugates) on the unit circle. Two of these Floquet multiplier conjugate pairs have infinite multiplicity on the infinite lattice.  As $\lambda$ increases from zero, these pairs spread out into bands along the unit circle as shown in Figure \ref{fig4}. At approximately $\lambda=0.42,$ a lone pair of conjugate Floquet multipliers emerges out of the bands into a small gap centered at $+1.$ These multipliers then collide at $+1$ when $\lambda\approx 0.45$ and move off the unit circle onto the real-axis, indicating a transition to dynamic instability. Numerically, this transition value appears to be independent of the lattice size $M$ and the temporal resolution $N$. In addition, any small departures from the unit circle along the spectral bands in Figure \ref{fig4} are found to decay with lattice size $M$. Such finite-size instabilities have been shown to disappear on the infinite lattice, as discussed in \cite{marin}. Thus, Figure \ref{fig4} indicates a large parameter window in $\lambda$ where midgap in- and out-of-phase breathers are linearly spectrally stable.

\section{Nonlinear plane waves and their stability}\label{sec:plane}

To further elucidate breather properties in the bandgap, we derive exact \textit{de-localized} nonlinear plane wave solutions to \eqref{eq1} and analyze their spectral stability.  Specifically, we are interested in nonlinear plane waves bifurcating from the trivial solution at the band-edges of the phonon spectrum and into the gap as shown in Figure \ref{fig2} (here we do not assume that the gap is necessarily small as in section \ref{sec:pde}).  We first analyze the stability of these states by using perturbation theory in the weakly nonlinear regime and then numerically compute their Floquet spectra in both the perturbative and non-perturbative regimes.  

We find that for a hardening (softening) nonlinearity in the infinite lattice, the in-phase (out-of-phase) nonlinear plane wave bifurcating from the top (bottom) edge of the acoustic (optical) band into the gap becomes dynamically \textit{unstable} to perturbative waves with wavevectors concentrated around $\mathbf{M}$.  It has been hypothesized in \cite{kastner,flach2} that the instability of band-edge plane waves in this regime gives rise to discrete breathers with arbitrarily small nonzero amplitudes, provided they exist. The energy thresholds for the formation of discrete breathers depend on both the dimension of the lattice and the highest exponent in the power-law potential function~\cite{We:99}.  Numerical studies suggest that both the localization strength and amplitude of these breathers approach zero as their frequency approaches the edge of the band \cite{Hof:23}.  We expect that the continuum theory accurately describes breather dynamics in this regime (see section \ref{sec:pde}). 

However, not all breathers behave this way near the phonon spectrum. A generic subclass of breathers is known to maintain nonzero amplitudes even when their frequency approaches a band-edge.  It has been hypothesized that this phenomenon is directly related to the lack of a bifurcating band-edge plane wave or, if one exists, its dynamic stability~\cite{kastner}. For the lattice model considered here, this argument suggests that small-amplitude discrete breathers do not exist near one edge of the gap. Indeed, this is consistent with the asymmetric limiting behaviors of gap solitons stated in Theorem \ref{gapsol} near the two band-edges. 

We remark that an analogous approach to our analysis of nonlinear plane waves in the lattice was given in the context of the continuous NLS equation to study nonlinear Bloch waves and their connection to gap solitons in periodic systems in \cite{zhang}.  A rigorous analysis of the bifurcation curves of nonlinear Bloch waves in the Gross-Pitaevskii equation with a periodic potential was given in \cite{dohnal2}. The authors of \cite{zhang, dohnal2} also study the continuation of breathers from gaps into spectral bands and give a description of these delocalized hybrid states or ``out-of-gap" solitons and their oscillatory tails.  

\subsection{Nonlinear plane waves}
In this section, we first consider a lattice with a finite and even number of unit cells, $n,m=0,1,\cdots S-1$, like the one shown in Figure \ref{fig3}(a) and with periodic boundary conditions.  We then consider the lattice in the thermodynamic limit as $S\to\infty.$

We define the discrete Fourier transform on the lattice
\begin{equation}
\begin{pmatrix}
    \hat{x}_{\mathbf{k}}^A \\
    \hat{x}_{\mathbf{k}}^B
\end{pmatrix}=\dfrac{1}{S^2}\sum_{n,m=0}^{S-1}\begin{pmatrix}
    x_{n,m}^A \\
    x_{n,m}^B
\end{pmatrix}e^{i\left(n\mathbf{a}_1+m\mathbf{a}_2\right)\cdot \mathbf{k}}
\label{forward}
\end{equation}
for discrete $\mathbf{k}=S^{-1}\left(l\mathbf{b_1}+p\mathbf{b}_2\right)$ modulo $\mathcal{B},$ where $\mathbf{b}_1, \mathbf{b}_2$ are the reciprocal vectors of the honeycomb lattice and $l,p$ are integers in $\{-S/2+1,\cdots,S/2\}.$ We denote this equally spaced partition of the Brillouin zone by the $S^2$ wavevectors as $\mathcal{B}_{S}.$  

The inverse transform is then given by
\begin{equation}
\begin{pmatrix}
    x_{n,m}^A \\
    x_{n,m}^B
\end{pmatrix}=\sum_{\mathbf{k}\in\mathcal{B}_{S}}\begin{pmatrix}
    \hat{x}_{\mathbf{k}}^A \\
    \hat{x}_{\mathbf{k}}^B
\end{pmatrix}e^{-i\left(n\mathbf{a}_1+m\mathbf{a}_2\right)\cdot \mathbf{k}}.
\label{backward}
\end{equation}

Substituting \eqref{backward} into equation \eqref{eq1} and transforming the result using \eqref{forward} gives the equivalent equation in terms of the Fourier coefficients 
\begin{align}
\label{fourier}
&\ddot{\hat{x}}_{\mathbf{k}}^A=\lambda e^{i\mathbf{a}_2\cdot {\mathbf{k}}}\left(1+e^{-i\mathbf{a}_1\cdot {\mathbf{k}}}\right)\hat{x}_{\mathbf{k}}^B\\ \nonumber
&\text{-}\dfrac{1}{S^2}\sum_{n,m=0}^{S-1}\partial_{z_1}V\left(\sum_{{\mathbf{k}}'\in\mathcal{B}_{S}}\hat{x}^A_{{\mathbf{k}}'}e^{-i(\mathbf{a}_1n+\mathbf{a}_2m)\cdot {\mathbf{k}}'},\sum_{{\mathbf{k}}''\in\mathcal{B}_{S}}\hat{x}^B_{{\mathbf{k}}''}e^{-i(\mathbf{a}_1n+\mathbf{a}_2m)\cdot {\mathbf{k}}''}\right)e^{i(\mathbf{a}_1n+\mathbf{a}_2m)\cdot {\mathbf{k}}}\\ \nonumber
&\ddot{\hat{x}}_{\mathbf{k}}^B=\lambda e^{-i\mathbf{a}_2\cdot {\mathbf{k}}}\left(1+e^{i\mathbf{a}_1\cdot {\mathbf{k}}}\right)\hat{x}_{\mathbf{k}}^A\\ \nonumber
&\text{-}\dfrac{1}{S^2}\sum_{n,m=0}^{S-1}\partial_{z_2}V\left(\sum_{{\mathbf{k}}'\in\mathcal{B}_{S}}\hat{x}^A_{{\mathbf{k}}'}e^{-i(\mathbf{a}_1n+\mathbf{a}_2m)\cdot {\mathbf{k}}'},\sum_{{\mathbf{k}}''\in\mathcal{B}_{S}}\hat{x}^B_{{\mathbf{k}}''}e^{-i(\mathbf{a}_1n+\mathbf{a}_2m)\cdot {\mathbf{k}}''}\right)e^{i(\mathbf{a}_1n+\mathbf{a}_2m)\cdot{\mathbf{k}}}.
\end{align}

This equation is greatly simplified at the edges of the phonon spectrum 
\eqref{eq2} at $\mathbf{k}=\mathbf{\Gamma}$ and $\mathbf{k}=\mathbf{M}$. Here, we are interested in the latter case and take $\hat{x}_{\mathbf{k}}^{A,B}=\hat{x}_{M}^{A,B}\delta_{\mathbf{M}},$ where $\delta_{\mathbf{M}}$ is the delta function centered on the wavevector $\mathbf{M}.$  

For the potential function in \eqref{potential}, equation \eqref{fourier} simplifies to
\begin{align}
\label{m}
&\ddot{\hat{x}}_{M}^A=-\omega_0^2\hat{x}_{M}^A+\left(1-2\lambda\right)\hat{x}_{M}^B-g\left[\left(\hat{x}_M^A\right)^2+\left(\hat{x}_M^B\right)^2\right]\hat{x}_M^A\\ \nonumber
&\ddot{\hat{x}}_{M}^B=-\omega_0^2\hat{x}_{M}^B+\left(1-2\lambda\right)\hat{x}_{M}^A-g\left[\left(\hat{x}_M^A\right)^2+\left(\hat{x}_M^B\right)^2\right]\hat{x}_M^B.
\end{align}
We remark that the linear part of system \eqref{m} decouples at the semi-Dirac point when $\lambda=\lambda_*$ and has a degenerate eigenvalue of $\omega_0^2.$  Like the anti-continuum problem in section \ref{sec:anti}, for $g=\pm 1$, there are in-phase and out-of-phase even-in-time periodic solutions for initial data $\hat{x}_{M}^A(0)=\pm\hat{x}_M^B(0)=a$ and $\dot{\hat{x}}_{M}^A(0)=\dot{\hat{x}}_{M}^B(0)=0$, respectively.   

For $\lambda\in[0,\lambda_*),$ the phonon spectrum has a bandgap of width $2(1-2\lambda)$ (see Figure \ref{fig2}).  When $g=1,$ the in-phase nonlinear plane-wave bifurcates from the zero solution at the top edge of the acoustic band where $\omega_-^2(\lambda):=\omega_-^2(\mathbf{M};\lambda)=\omega_0^2-1+2\lambda$ into the gap.  Similarly, when $g=-1,$ the out-of-phase nonlinear plane-wave bifurcates from the zero solution at the bottom edge of the optical band where $\omega_+^2(\lambda):=\omega_+^2(\mathbf{M};\lambda)=\omega_0^2+1-2\lambda$ into the gap.

These nonlinear plane-waves are given explicitly by
\begin{equation}\label{plane1}
    \begin{pmatrix}
    x_{n,m}^{A}\\
    x_{n,m}^B
\end{pmatrix}^{\pm}=\begin{pmatrix}
    \hat{x}_M^A \\
    \hat{x}_M^B
\end{pmatrix}^{\pm}e^{-i\pi m}=P_{\pm}(t;\lambda)\begin{pmatrix}1\\
\pm 1
\end{pmatrix}e^{-i\pi m},
\end{equation}
where $P_{\pm}(t;\lambda)$ is the cnoidal function in \eqref{cn} with $\Omega_{\pm}$ replaced by $\omega_{\pm}$. We note that when $\lambda=0$, $\omega_{\pm}=\Omega_{\pm}$ and $P_{\pm}$ is identical to the anti-continuum limit solution in section \ref{sec:anti}, but the solution in \eqref{plane1} is nonzero over the entire lattice.

To investigate the dynamic stability of the nonlinear plane waves in \eqref{plane1}, we linearize system \eqref{fourier} about their orbit.  Due to the fact that the potential \eqref{potential} is even symmetric, we obtain
\begin{equation}\label{lin}
\begin{pmatrix}
    \ddot{y}_{\mathbf{k}}^A \\
    \ddot{y}_{\mathbf{k}}^B
\end{pmatrix}=
\left[-\mathcal{H}_{V}[\hat{X}^{\pm}(t)]+
\begin{pmatrix}
    0 & \lambda e^{i\mathbf{k}\cdot \mathbf{a}_2}(1+e^{-i\mathbf{k}\cdot \mathbf{a}_1})  \\
     \lambda e^{-i\mathbf{k}\cdot \mathbf{a}_2}(1+e^{i\mathbf{k}\cdot \mathbf{a}_1}) & 0
\end{pmatrix}\right]
\begin{pmatrix}
    y_{\mathbf{k}}^A \\
    y_{\mathbf{k}}^B
\end{pmatrix},
\end{equation}
where $\mathcal{H}_{V}[\hat{X}^{\pm}(t)]$ is the Hessian of $V$ evaluated at $\hat{X}^{\pm}:=\left(\hat{x}_M^{A},\hat{x}_M^B\right)^{\pm\top}$ in \eqref{plane1}. Equation \eqref{lin} is a linear block-diagonal system of $4\Lambda^2$ second-order differential equations with periodic coefficients~\cite{hill}.

\subsection{Perturbative stability analysis of band edge plane-waves}

To study the stability of the nonlinear plane-waves near the edges of the bandgap, we first derive the leading-order expression to approximate $P_{\pm}(t;\lambda)$ for small amplitude waves, $|a|\ll 1$, by using the Poincar\'{e}-Lindstedt method~\cite{verhulst}. The result is
\begin{equation}\label{per_1}
    P_{\pm}(t;\lambda)=a\cos\left[\left(\omega_{\mp}\pm\dfrac{3a^2}{4\omega_{\mp}}+O\left(a^4\right)\right)t\right]+O\left(a^3\right).
\end{equation}

In Floquet theory~\cite{hill}, it is well known that a transition from linear stability to instability of the periodic solution in \eqref{plane1} occurs when it is perturbed by a periodic orbit of the same period.  At this point a \textit{tangent bifurcation} occurs when two Floquet multipliers collide at $+1$ and leave the unit circle (i.e., see Figure \ref{fig4}). Following the arguments in \cite{flach2, kastner}, we locate this transition to leading order in $a$ by using the Poincar\'{e}-Lindstedt method again on the linearized system \eqref{lin} about \eqref{per_1}. After obtaining its frequency expansion, $\omega(a)$, we equate it with the frequency expansion in \eqref{per_1} and solve for the critical value, $a_c.$  We are interested in values of $a_c$ in the thermodynamic limit.

Introducing the perturbation expansions in powers of $a^2$:
\[y_{\mathbf{k}}^{A,B}=y_{\mathbf{k},(0)}^{A,B}+y_{\mathbf{k},(1)}^{A,B}a^2+\cdots\quad\text{ and }\quad\omega^2(a)=\omega_{(0)}^2+\omega_{(1)}^2a^2+\cdots, \]
we define the scaled time variable $\tau=\omega(a)t.$  Plugging the expansions into \eqref{lin} gives at zeroth-order:
\[L_0\mathbf{Y_{\mathbf{k},(0)}}:=\begin{pmatrix}
    \omega_{(0)}^2\dfrac{d^2}{d\tau^2}+\omega_0^2 & -1-\lambda e^{i\mathbf{k}\cdot \mathbf{a}_2}(1+e^{-i\mathbf{k}\cdot \mathbf{a}_1})  \\
     -1-\lambda e^{-i\mathbf{k}\cdot \mathbf{a}_2}(1+e^{i\mathbf{k}\cdot \mathbf{a}_1}) & \omega_{(0)}^2\dfrac{d^2}{d\tau^2}+\omega_0^2
\end{pmatrix}\begin{pmatrix}
    y_{\mathbf{k},(0)}^A \\
    y_{\mathbf{k},(0)}^B
\end{pmatrix}=0.\]

The general solution to the zeroth-order equation can be expressed as
\begin{align*}
    \begin{pmatrix}
    y_{\mathbf{k},(0)}^A \\
    y_{\mathbf{k},(0)}^B
\end{pmatrix}=&\left[\alpha\exp\left(\dfrac{i\omega_{-}(\mathbf{k};\lambda)\tau}{\omega_{(0)}}\right)+\beta\exp\left(-\dfrac{i\omega_{-}(\mathbf{k};\lambda)\tau}{\omega_{(0)}}\right)\right]\begin{pmatrix}v_1(\mathbf{k};\lambda)  \\ v_2(\mathbf{k};\lambda) \end{pmatrix}\\ \nonumber
&+\left[\gamma\exp\left(\dfrac{i\omega_{+}(\mathbf{k};\lambda)\tau}{\omega_{(0)}}\right)+\zeta\exp\left(-\dfrac{i\omega_{+}(\mathbf{k};\lambda)\tau}{\omega_{(0)}}\right)\right]\begin{pmatrix}-v_1(\mathbf{k};\lambda) \\ v_2(\mathbf{k};\lambda) \end{pmatrix},
\end{align*}
where $\omega^2_{\mp}(\mathbf{k};\lambda)$ are the acoustic/optical bands in \eqref{eq2}, $\alpha,\beta,\gamma,\zeta$ are complex constants, and $\left(\pm v_1(\mathbf{k};\lambda),v_2(\mathbf{k};\lambda)\right)^{\top}$ are the corresponding eigenvectors.  The components $v_1$ and $v_2$ are nonvanishing for all $\mathbf{k}\in\mathcal{B}_{S}$ and $v_1=v_2=1$ at both $\mathbf{k}=\mathbf{M}$ and $\mathbf{k}=\mathbf{\Gamma}$ for all $0\leq\lambda<\lambda_*$.  

At first-order, we obtain
\[L_0\begin{pmatrix}
    y_{\mathbf{k},(1)}^A \\
    y_{\mathbf{k},(1)}^B
\end{pmatrix}=\begin{pmatrix}
    \omega_{(1)}^2\dfrac{d^2}{d\tau^2}+4g\cos^2(\tau) & 2\cos^2(\tau)  \\
     2\cos^2(\tau) & \omega_{(1)}^2\dfrac{d^2}{d\tau^2}+4g\cos^2(\tau)
\end{pmatrix}\begin{pmatrix}
    y_{\mathbf{k},(0)}^A \\
    y_{\mathbf{k},(0)}^B
\end{pmatrix}.
\]
The secular terms in the first-order equation can be removed only when the magnitudes of the components $v_1$ and $v_2$ are equal, which occurs only at the edges of the bands.  This is true since the one free parameter, $\omega_{(1)}^2$, must be chosen to eliminate the terms with $\cos(\tau)$ in \textit{both} equations above. This results in the condition $v_1^2=v_2^2$, with the relative sign chosen to match the corresponding eigenvector.  

Choosing $\omega_{(0)}^2=\omega_{\mp}^2(\mathbf{k};\lambda)$ and $\omega_{(1)}^2=\pm9/2$ for $g=\pm1$,
respectively, removes the secular terms and gives a leading-order solution that is bounded for all time.  The Floquet multiplier corresponding to these periodic orbits remains fixed at $+1$ and does not lead to a tangent bifurcation~\cite{flach2}.  However, in the thermodynamic limit, equating the leading order frequency expansion, $\omega^2(a)$, with the nonlinear plane wave in \eqref{per_1}, gives 
\begin{align}\label{ac}
   &a_{c,M}^2:=\dfrac{1}{3}\left(\omega^2_{\mp}(\lambda)-\omega_{\mp}^2(\mathbf{k};\lambda)\right)\to0\quad\text{ as }\mathbf{k}\to\mathbf{M},\\
   \nonumber
   &a_{c,\Gamma}^2:=\dfrac{1}{3}\left(\omega^2_{\mp}(\lambda)-\omega_{\mp}^2(\mathbf{k};\lambda)\right)\to \dfrac{4\lambda}{3}\quad\text{ as }\mathbf{k}\to\mathbf{\Gamma}.
\end{align}
This is due to the fact that the phonon bands and the eigenvectors are continuous functions of the wavevector $\mathbf{k},$ which densely fills the Brillouin zone as $S\to\infty.$ 

The first limit in \eqref{ac} is valid for all $0\leq\lambda<\lambda_*$.  This suggests that \eqref{per_1} undergoes a first transition to instability near the wavevector $\mathbf{M}$ as it bifurcates from the zero solution and into the bandgap. The second limit in \eqref{ac} is valid as long as $\lambda$ is small enough to safely neglect higher-order terms in the perturbation expansion.   In this case, the nonlinear plane wave undergoes another tangent bifurcation near $\mathbf{\Gamma}$ at the nonzero amplitude $2\sqrt{\lambda/3}.$  We find excellent agreement between these perturbative results and the numerics presented below.  We remark that when $\lambda=0$ the unit cells of the lattice are uncoupled and the phonon bands are perfectly flat over the entire Brillouin zone.  In this case, \eqref{plane1} is dynamically stable and nearly all of the system's Floquet multipliers are repeated, located exactly at $+1$.

\subsection{Numerical stability of nonlinear plane waves}

We now numerically compute the Floquet spectra of the linearized system \eqref{lin} about the exact nonlinear plane wave in \eqref{plane1} without any perturbative assumption on its amplitude.  
 \begin{figure}[h!]
  \centering
  \includegraphics[scale=0.65]{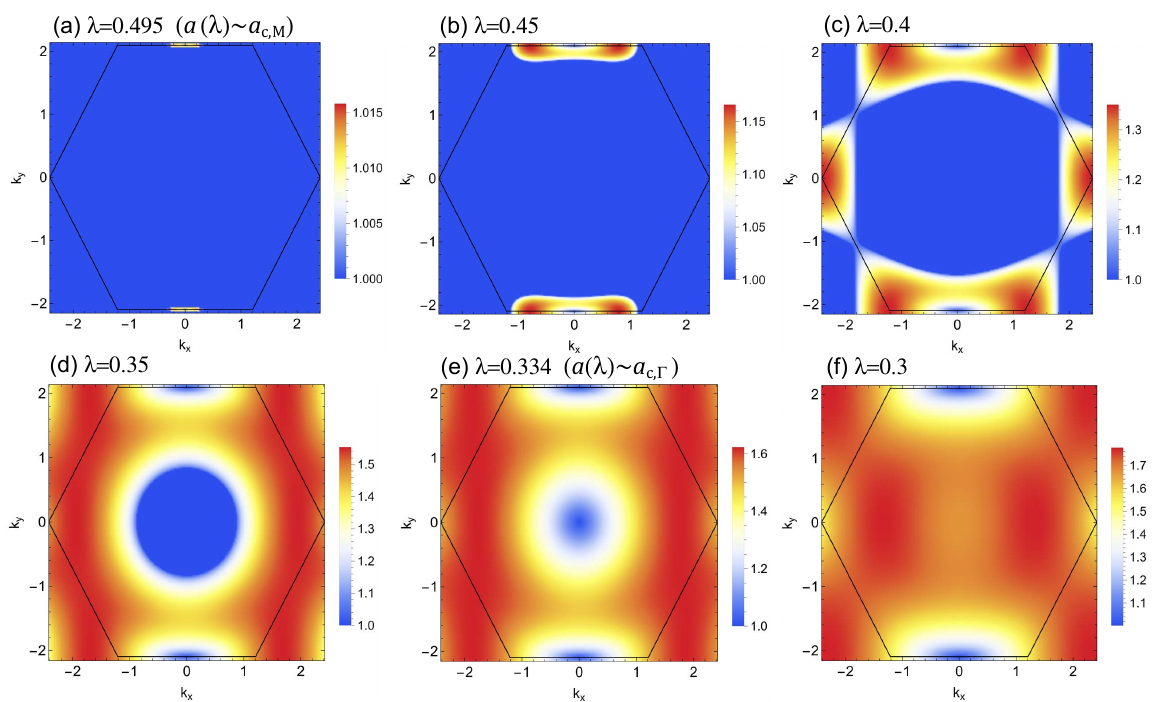}
  \vspace{-0.2cm}
  \caption{Magnitude of the Floquet spectra of the hardening nonlinear plane wave in \eqref{plane1} over the Brillouin zone, depicted by black lines, with $S=250$. Regions where the magnitude equals one are stable and those greater than one are unstable. (a) shows when the gap is small and the band-edge plane wave has undergone a tangent bifurcation near the wavevector $\mathbf{M}$. (b)-(f) show that as the gap widens there are growing regions of instability and a persistent neighborhood of stability around $\mathbf{M}.$ (e) shows when $a^2(\lambda)\approx a_{c,\Gamma}^2=4\lambda/3$, where \eqref{plane1} undergoes another tangent bifurcation near the $\Gamma$-point.  (f) shows just after the tangent bifurcation at $a\gtrsim a_{c,\Gamma}$.  }
  \label{fig5}
\end{figure}
\noindent Here, we take $g=1$ and consider the in-phase wave in \eqref{plane1} that emerges from the upper-edge plane-wave of the acoustic band at $\mathbf{k}=\mathbf{M}.$   We compute the Floquet spectrum at the point where the bifurcation curve has crossed the gap and, hence, the frequency of \eqref{plane1} intersects the lower-edge of the optical band.  Figure \ref{fig5} shows the magnitude of the largest Floquet multiplier for each $\mathbf{k}\in\mathcal{B}_{S}$ using the numerical solution to \eqref{lin}. Stable regions are colored blue.

Panel (a) in Figure \ref{fig5} shows that the bottom band-edge plane wave undergoes a tangent bifurcation near $\mathbf{M}$ when the gap is small ($\lambda=0.495\lesssim\lambda_*=1/2$) and so the amplitude of \eqref{plane1}, denoted by $a(\lambda)$, is slightly greater than $a_{c,M}=0$.  This is consistent with the result of our perturbation analysis in \eqref{ac}.  As the gap widens, $\lambda$ decreases and $a(\lambda)$ increases.  The subsequent panels in \ref{fig5} show a growing region of instability emanating from $\mathbf{M}$ with a persistent neighborhood of stability centered on $\mathbf{M}$.  This behavior is akin to the stability/instability patterns of Arnold tongues~\cite{arnold}.   Panels (e) and (f) show that a second tangent bifurcation occurs at $\mathbf{\Gamma}$ when $a(\lambda)\approx a_{c,\Gamma}$ predicted by our perturbation theory in \eqref{ac}.

Figure \ref{fig5} is consistent with the different behavior of discrete breathers near the lower and upper edges of a small gap, as described by the gap line-solitons in Theorem \ref{gapsol}.  For a fixed nonzero bandgap, weakly-localized (in $Z$), small-amplitude discrete breathers emerge near the bottom edge of the gap due to the instability of the lower band-edge plane wave~\cite{flach2,kastner}.  In contrast, no small-amplitude discrete breathers exist near the upper edge of the gap due to the lack of a band-edge plane wave that ``repels" from the linear spectrum for a hardening nonlinearity.  Indeed, there is a mismatch between the in-phase dynamics of the continued breather and the out-of-phase dynamics of the plane wave at the edge of the optical band.

\section{Concluding remarks}\label{sec:con}
We conducted a formal asymptotic, numerical, and stability analysis of discrete breathers and nonlinear plane waves on a honeycomb lattice near a unique point in the dispersion relation known as semi-Dirac crossing. Our analysis connects breather dynamics in the highly localized and strongly nonlinear regime with breather dynamics in the long-wave and weakly nonlinear regime. Such an approach gives a complete description of hybrid localized structures with spatially extended tails described by exact separable solutions to an asymptotic effective PDE theory. Even though an existence proof for global localized solutions to the 2D semi-Dirac equation is lacking, we do not require them here to accurately describe a large family of long-lived discrete breather states on the lattice.  Our work provides a framework for describing coherent structures in more complicated lattice models of current interest in condensed matter systems and metamaterial design, including 2D layered and twisted lattices~\cite{twist,roadmap}.

\bibliographystyle{siamplain}
\bibliography{references}

\begin{thebibliography}{10}

\bibitem{aceves}
{\sc A.~B. Aceves and S.~Wabnitz}, {\em Self induced transparency solitons in nonlinear refractive periodic media}, Phys. Lett. A, 141 (1989), p.~3742.

\bibitem{arnold}
{\sc V.~I. Arnol'd}, {\em Remarks on the perturbation theory for problems of {M}athieu type}, Russ. Math. Surv., 38 (1983).

\bibitem{bertoldi}
{\sc K.~Bertoldi, V.~Vitelli, J.~Christensen, and M.~V. Hecke}, {\em Flexible mechanical metamaterials}, Nature Reviews Materials, 2 (2017), p.~17066.

\bibitem{twist}
{\sc A.~Bistritzer and A.~H. MacDonald.}, {\em {M}oiré bands in twisted double-layer graphene}, Proceedings of the National Academy of Sciences, 108 (2011).

\bibitem{borelli}
{\sc W.~Borelli}, {\em Stationary solutions for the 2d critical {D}irac equation with {K}err nonlinearity}, Journal of Differential Equations, 263 (2017).

\bibitem{caz}
{\sc T.~Cazenave and L.~Vazquez}, {\em Existence of localized solutions for a classical nonlinear {D}irac field}, Communications in Mathematical Physics, 105 (1986).

\bibitem{haller}
{\sc R.~Cruickshank, F.~Lorenzi, A.~L. Rooij, E.~F. Kerr, T.~Hilker, S.~Kuhr, L.~Salasnich, and E.~Haller}, {\em Experimental observation of single- and multisite matter-wave solitons in an optical accordion lattice}, Physical Review Letters, 135 (2025), p.~263404.

\bibitem{dai}
{\sc Y.~Dai, H.~Yu, Z.~Zhu, Y.~Wang, and L.~Huang}, {\em Discrete breathers and energy localization in a nonlinear honeycomb lattice}, Physical Review E, 104 (2021), p.~064201.

\bibitem{roadmap}
{\sc F.~J.~G. de~Abajo, D.~N. Basov, F.~H.~L. Koppens, et~al.}, {\em Roadmap for photonics with 2d materials}, ACS Photonics, 12 (2025).

\bibitem{dietl}
{\sc P.~Dietl, F.~Pi\`{e}chon, and G.~Montambaux}, {\em New magnetic field dependence on {L}andau levels in a graphenelike structure}, Physical Review Letters, 100 (2008), p.~236405.

\bibitem{dohnal}
{\sc T.~Dohnal and A.~Aceves}, {\em Optical soliton bullets in (2+1)d nonlinear {B}ragg resonant periodic geometries}, Studies in Applied Mathematics, 115 (2005), pp.~209--232.

\bibitem{dohnal2}
{\sc T.~Dohnal and H.~Uecker}, {\em Bifurcation of nonlinear {B}loch waves from the spectrum in the {G}ross–{P}itaevskii equation}, Journal of Nonlinear Science, 26 (2016), pp.~581--618.

\bibitem{sere}
{\sc M.~J. Esteban and E.~S\'{e}r\'{e}}, {\em Stationary states of the nonlinear {D}irac equation: A variational approach}, Communications in Mathematical Physics, 171 (1995).

\bibitem{feff14}
{\sc C.~L. Fefferman, J.~P. Lee-Thorp, and M.~I. Weinstein}, {\em Topologically protected states in one-dimensional continuous systems and {D}irac points}, Proceedings of the National Academy of Sciences,  (2014).

\bibitem{fibich}
{\sc G.~Fibich}, {\em The Nonlinear Schrödinger Equation: Singular Solutions and Optical Collapse}, Springer, 2015.

\bibitem{flach2}
{\sc S.~Flach}, {\em Tangent bifurcation of band edge plane waves, dynamical symmetry breaking and vibrational localization}, Physica D: Nonlinear Phenomena, 91 (1996).

\bibitem{flasch08}
{\sc S.~Flach and A.~V. Gorbach}, {\em Discrete breathers - advances in theory and applications}, Physics Reports, 467 (2008).

\bibitem{sigal}
{\sc P.~D. Hislop and I.~M. Sigal}, {\em The Essential Spectrum: Weyl’s Criterion}, Springer, 1995.

\bibitem{Hof:23}
{\sc A.~Hofstrand, H.~Li, and M.~I. Weinstein}, {\em Discrete breathers of nonlinear dimer lattices: bridging the anti-continuous and continuous limits}, Journal of Nonlinear Science, 33 (2023).

\bibitem{IW10}
{\sc B.~Ilan and M.~Weinstein}, {\em Band-edge solitons, nonlinear {S}chr\"odinger/{G}ross–{P}itaevskii equations, and effective media}, Multiscale Modeling and Simulation, 8 (2010).

\bibitem{ess}
{\sc G.~Jotzu, M.~Messer, R.~Desbuquois, M.~Lebrat, T.~Uehlinger, D.~Greif, and T.~Esslinger}, {\em Experimental realization of the topological {H}aladane model with ultracold fermions}, Nature, 515 (2014), pp.~237--240.

\bibitem{kane}
{\sc C.~L. Kane and E.~J. Mele}, {\em Quantum spin {H}all effect in graphene}, Physical Review Letters, 95 (2005).

\bibitem{kastner}
{\sc M.~Kastner}, {\em Energy thresholds for discrete breathers}, Physical Review Letters, 92 (2004), p.~104301.

\bibitem{kouk2}
{\sc V.~Koukouloyannis, P.~G. Kevrekidis, K.~J.~H. Law, I.~Kourakis, and D.~J. Frantzeskakis}, {\em Existence and stability of multisite breathers in honeycomb and hexagonal lattices}, Journal of Physics A: Mathematical and Theoretical, 43 (2010).

\bibitem{kouk}
{\sc V.~Koukouloyannis and I.~Kourakis}, {\em Discrete breathers in hexagonal dusty plasma lattices}, Physical Review E, 80 (2009), p.~026402.

\bibitem{Li:24}
{\sc H.~Li, A.~Hofstrand, and M.~I. Weinstein}, {\em Stability of traveling waves in a nonlinear hyperbolic system approximating a dimer array of oscillators}, Journal of Nonlinear Science, 35 (2025).

\bibitem{ma}
{\sc T.~Ma, A.~B. Khanikaev, S.~H. Mousavi, and G.~G. Shvets}, {\em Electromagnetic waves around sharp corners: topologically protected photonic transport in metawaveguides.}, Physical Review Letters, 114 (2015).

\bibitem{hill}
{\sc W.~Magnus and S.~Winkler}, {\em Hill's Equation}, Dover, 1979.

\bibitem{sax}
{\sc J.-C. Maraver, N.~Boussaid, A.~Comech, P.~G. Kevrekidis, and A.~Saxena}, {\em Solitary waves in the nonlinear {D}irac equation}, Nonlinear Systems, 1 (2018), \url{https://doi.org/https://doi.org/10.1007/978-3-319-66766-9_4}.

\bibitem{ma96}
{\sc J.~L. Marin and S.~Aubry}, {\em Breathers in nonlinear lattices: numerical calculation from the anticontinous limit}, Nonlinearity, 9 (1996).

\bibitem{marin}
{\sc J.~L. Marin and S.~Aubry}, {\em Finite size effects on instabilities of discrete breathers}, Physica D, 119 (1998).

\bibitem{nirenberg}
{\sc L.~Nirenberg}, {\em Topics in Nonlinear Functional Analysis}, American Mathematical Society, 2001.

\bibitem{novo}
{\sc K.~S. Novoselov, A.~K. Geim, S.~V. Morosov, D.~Jiang, Y.~Zhang, S.~V. Dobonos, J.~V. Grigorieva, and A.~A. Firsov}, {\em Electric field effect in atomically thin carbon films}, Science, 306 (2004).

\bibitem{palermo}
{\sc F.~Palermo, L.~Q. English, J.~C. P.-Maraver, and P.~G. Kevrekidis}, {\em Nonlinear edge modes in a honeycomb electrical lattice near the {D}irac points}, Physics Letters A, 384 (2020).

\bibitem{pelinovsky1}
{\sc D.~Pelinovsky}, {\em Survey on global existence in the nonlinear {D}irac equations in one spatial dimension}, Res. Inst. Math. Sci. (RIMS),  (2011), \url{https://doi.org/https://arxiv.org/pdf/1011.5925}.

\bibitem{pelinovsky2}
{\sc D.~Pelinovsky and Y.~Shimabukuro}, {\em Transverse instability of line solitary waves in massive {D}irac equations}, Journal of Nonlinear Science, 26 (2016).

\bibitem{basov}
{\sc Y.~Shao, S.~Moon, A.~N. Rudenko, J.~Wang, J.~Herzog-Arbeitman, M.~Ozerov, D.~Graf, Z.~Sun, R.~Queiroz, S.~H. Lee, Y.~Zhu, Z.~Mao, M.~I. Katsnelson, B.~A. Bernevig, D.~Smirnov, A.~J. Millis, and D.~N. Basov}, {\em Semi-{D}irac fermions in a topological metal}, Physical Review X, 14 (2024), p.~041057.

\bibitem{hibbins}
{\sc T.~A. Starkey, V.~Kyrimi, G.~P. Ward, J.~R. Sambles, and A.~P. Hibbins}, {\em Experimental characterization of the bound acoustic surface modes supported by honeycomb and hexagonal hole arrays}, Scientific Reports, 9 (2019), p.~15773.

\bibitem{ssh79}
{\sc W.~P. Su, J.~R. Schrieffer, and A.~J. Heeger}, {\em Solitons in polyacetylene}, Physical Review Letters, 42 (1979), p.~1698, \url{https://doi.org/https://doi.org/10.1103/PhysRevLett.42.1698}.

\bibitem{thirring}
{\sc W.~E. Thirring}, {\em A soluble relativistic field theory}, Annals of Physics, 3 (1958), pp.~91--112.

\bibitem{verhulst}
{\sc F.~Verhulst}, {\em Nonlinear Differential Equations and Dynamical Systems}, Springer, 1996.

\bibitem{wattis}
{\sc J.~A.~D. Wattis and L.~M. James}, {\em Discrete breathers in honeycomb {F}ermi–{P}asta–{U}lam lattices}, J. Phys. A: Math. Theor., 47 (2014), p.~345101.

\bibitem{We:99}
{\sc M.~I. Weinstein}, {\em Excitation thresholds for nonlinear localized modes on lattices}, Nonlinearity, 12 (1999).

\bibitem{zhang}
{\sc Y.~Zhang and B.~Wu}, {\em Composition relation between gap solitons and {B}loch waves in nonlinear periodic systems}, Physical Review Letters, 102 (2009), p.~093905.

\end{thebibliography}
\end{document}